\documentclass[conference]{IEEEtran}

\usepackage{caption}
\usepackage{subcaption}
\usepackage[utf8]{inputenc}
\usepackage[T1]{fontenc}
\usepackage[english]{babel}
\usepackage{csquotes}
\usepackage{cite}

\usepackage[textsize=footnotesize]{todonotes} 
\usepackage[shortlabels]{enumitem}
\usepackage{amsmath}
\usepackage{amsthm}
\usepackage{mathtools}
\usepackage{stmaryrd}
\usepackage{amssymb}
\usepackage{mathrsfs}
\usepackage{tikz}
\usetikzlibrary{shapes}
\usepackage[colorlinks]{hyperref}
\usepackage[nameinlink]{cleveref}
\usepackage{ebproof} 

\usepackage{tikz-cd}
\tikzcdset{row sep/abysmal/.initial=0.1em} 
\tikzcdset{row sep/none/.initial=-0.4em}
\usepackage{quiver}

\usepackage[switch]{lineno}	

\renewenvironment{proof}{\begin{IEEEproof}}{\end{IEEEproof}}
\renewcommand{\IEEEQED}{\IEEEQEDopen}

\usepackage{notations}

\theoremstyle{definition}
\newtheorem{definition}{Definition}
\newtheorem{notation}{Notations}

\theoremstyle{plain}
\newtheorem{lemma}{Lemma}
\newtheorem*{lemma*}{Lemma}
\newtheorem{cor}{Corollary}

\newtheorem{theorem}{Theorem}
\newtheorem*{theorem*}{Theorem}

\newtheorem{proposition}{Proposition}
\newtheorem*{proposition*}{Proposition}
\newtheorem{fact}{Fact}

\theoremstyle{remark}
\newtheorem{remark}{Remark}
\newtheorem{example}{Example}

\setlist[enumerate,1]{label={(\arabic*)}}

\title{Compositional Taylor expansion in cartesian differential categories}

\author{
\IEEEauthorblockN{Walch Aymeric}
\IEEEauthorblockA{
\textit{Université Paris Cité, CNRS, IRIF}\\
F-75013 Paris, France \\
walch@irif.fr}}

\crefname{proposition}{\text{prop.}}{\text{props.}}
\Crefname{proposition}{\text{Prop.}}{\text{Props.}}
\crefname{definition}{\text{def.}}{\text{defs.}}
\Crefname{definition}{\text{Def.}}{\text{Defs.}}
\crefname{axiom}{\text{ax.}}{\text{ax.}}
\Crefname{axiom}{\text{Ax.}}{\text{Ax.}}
\crefname{theorem}{\text{thm.}}{\text{thms.}}
\Crefname{theorem}{\text{Thm.}}{\text{Thms.}}
\crefname{cor}{\text{cor.}}{\text{cor.}} 
\Crefname{cor}{\text{Cor.}}{\text{Cor.}}
\crefname{remark}{\text{rem.}}{\text{rem.}}
\Crefname{remark}{\text{Rem.}}{\text{Rem.}}
\crefname{notation}{\text{not.}}{\text{not.}}
\Crefname{notation}{\text{Not.}}{\text{Not.}}
\crefname{section}{\text{sec.}}{\text{sec.}}
\Crefname{section}{\text{Sec.}}{\text{Sec.}}
\crefname{figure}{\text{fig.}}{\text{fig.}}
\Crefname{figure}{\text{Fig.}}{\text{Fig.}}

\crefname{equation}{}{}
\Crefname{equation}{}{}

\newcommand{\version}[2]{#2} 

\begin{document}
\maketitle

\begin{abstract}
    This paper provides a compositional approach to Taylor expansion, in the 
    setting of cartesian differential categories.
    Taylor expansion is captured here by a functor
    that generalizes the tangent bundle functor to higher order derivatives.
    The fundamental properties of Taylor expansion then boils down to naturality 
    equations that turns this functor into a monad. This monad provides a categorical
    approach to higher order dual numbers and the jet bundle construction 
    used in automated differentiation. 
\end{abstract}


The combination of the theory of the differential calculus with the theory 
of programming languages has seen a tremendous growth in the last decades,
most notably in the fields of automated differentiation (AD)~\cite{Baydin17}
and of the differential $\lambda$-calculus~\cite{EhrhardRegnier02}.
Both AD and the differential $\lambda$-calculus aim at computing the derivative
of a program in a compositional way, this compositionality 
is crucial to scale those methods to complex assemblies of programs.
Because of this interplay between derivatives
and composition, category theory provides a strong mathematical 
basis for both of those fields.
Categorical semantics provides critical proofs
methods of correctness of the AD algorithm\cite{Mazza21,Huot23}, and the 
differential lambda calculus is deeply related to the 
categorical semantics of Linear Logic (LL) from its 
very inception~\cite{EhrhardRegnier02,Manzonetto12}.
Among those categorical approaches, cartesian differential 
categories~\cite{Blute09} provide a direct 
axiomatization of derivatives in any cartesian category. As such, they 
serve as a framework to understand the differential calculus through the lenses 
of compositionality. 

The compositionality of the derivative is expressed by the
\emph{chain rule}: $\derive{(g \comp f)}
= \derive[f(x)]{g} \cdot \derive{f}$.
%
The issue of the chain rule is that it is not entirely compositional, because
the derivative $\derivenoarg{(g \comp f)}$ also depends on $f$, and not only on 
$\derivenoarg{g}$ and $\derivenoarg{f}$. 
For this reason, one often consider the \emph{tangent bundle} 
operator $\T$ that intuitively maps $f : X \arrow Y$
to the function $\T f : (x, u) \mapsto (f(x), \derive{f} \cdot u)$. This operator
exists in any cartesian differential category, and the chain rule boils down 
to a functoriality equation on $\T$: 
$\T(g \comp f) = \T g \comp \T f$. 
Furthermore, the other axioms of the differential calculus 
(such as the linearity of the derivative or the symmetry of the higher 
order derivatives) turn out to be equivalent to 
naturality equations~\cite{Cockett14,Walch23} 
that turn $\T$ into a \emph{monad}~\cite{Walch23} whose algebraic 
structure is similar to that of dual numbers widely used in 
AD.
This suggests that differentiation is an effect in the sense 
of Moggi~\cite{Moggi91}, further cementing the 
use of category theory as a strong mathematical 
foundation for the differential calculus.

\textbf{Contribution of the paper:}
This article provides a similar compositional 
approach to Taylor expansion, in any cartesian differential category.
%
%
The computation of the Taylor expansion of a composition $g \comp f$
is not straightforward because
it requires to compute the higher order derivatives
$\hod{(g \comp f)} n$.
These derivatives are given by the Fa\'a di Bruno 
formula~\cite{Fraenkel78} \cref{eq:faa-di-bruno} which is 
notoriously combinatorial.
We solve this issue by defining an operator $\Tn$ 
that maps $f : X \arrow Y$ to $\Tn f : X \times X^n
    \arrow Y \times Y^n$
such that the $i$-th component of $\Tn f(x, u_1, \ldots, u_n)$ 
provides the term of order $\varepsilon^i$ in the Taylor expansion of 
$f(x + \varepsilon u_1 + \cdots + \varepsilon^n u_n)$.
The order $n$ Taylor expansion of $f$
can be recovered from $\Tn f$ simply
by computing $\Tn f(x, u, 0, \ldots, 0)$, but generalizing 
this to any vector $(x, u_1, \ldots, u_n)$ is the key to turn this 
into a compositional operation such that 
$\Tn (g \comp f) = \Tn g \comp \Tn f$.
We prove that $\Tn$ is a functor, and satisfies naturality 
equations that are clear analogues of the naturality equations of 
the tangent bundle $\T$. These equations turn
$\Tn$ into a monad, whose algebraic structure is 
similar to that of higher order dual numbers recently used in AD~\cite{Szirmay20}.
We explore some consequences of these observations, such as expressing the 
partial Taylor expansion with regard to one variable in terms 
of monad strengths, or by combining multiple Taylor expansions using 
distributive laws~\cite{Beck69}.

\textbf{Proofs methods:}
The combinatorics of Taylor expansion usually
involves quite a lot of heavy lifting.
A key contribution of this article is to hide all this combinatorics behind 
categorical equations. In particular, 
we will heavily rely on functoriality and naturality equations on the iterated tangent 
bundle functor $\T^n f : X^{2^n} \arrow Y^{2^n}$ 
rather than direct style equations involving the 
higher order derivatives $\hod{f} n : X \times X^{n} \arrow Y$.
More precisely, we prove that $\T^n$ has a canonical monad 
structure that arises from a 
\emph{distributive law}~\cite{Beck69} of the monad $\T$ on itself.
%
Then, we express $\Tn f : X \times X^{n} \arrow Y \times Y^{n}$ using 
$\T^n f$ instead of $\hod{f}{n}$. To this end, we 
generalize to any cartesian differential category the notion of 
\emph{higher order directional derivative} of~\cite{Huang06}, 
that we tweak with coefficients in order to erase redundancy.
%
Then, it is quite straightforward to check that the 
functoriality and naturality equations on $\T^n$ induce 
similar equations on $\Tn$.
This process also yields a family of morphisms $(\Stree)_X : \Tn X \arrow \T^n X$ 
that turns out to be a monad morphism~\cite{Maranda66}, or monad 
transformer~\cite{Liang95}, from $\Tn$ to $\T^n$.
This gives an insight on the link between the 
higher order derivatives as studied in~\cite{Lemay18,Garner21}
and Taylor expansion.

\textbf{Links with the differential $\lambda$-calculus:}
%
The operator $\Tn$  is directly 
inspired from the \emph{coherent Taylor expansion} of~\cite{EhrhardWalch25}.
In this article, Ehrhard and Walch provide 
a new categorical semantics of the differential $\lambda$-calculus~\cite{EhrhardRegnier02}
and of Taylor expansion~\cite{Ehrhard08} that accounts 
for the deterministic nature of computation, as opposed to the previously known 
models of~\cite{Manzonetto12}. This determinism is modeled by a notion of partial sums 
that prevents the summation of incompatible information (such as 
summing \texttt{true} and \texttt{false} in the type of boolean). 
Taylor expansion is axiomatized in their work as a bimonad, 
whose monad structure is an infinitary counterpart of our monad $\Tn$.
There are however two main differences with our work. First, 
they deal with infinite partial sums which are inherently positives in 
the sense that $x+y = 0$ implies $x = y = 0$. 
This clashes with the cartesian differential categories arising from 
traditional analysis. More importantly, Ehrhard and Walch do not provide any explicit 
description of their Taylor expansion in terms of derivatives outside of examples.
This sharply contrasts with our work in which we explicitly define Taylor expansion
from the derivative, in any cartesian differential category.
We explore in \cref{sec:taylor-expansion-def,sec:infinitary-sums} the relation 
between their Taylor expansion and ours.
In particular, we use our construction $\Tn$ 
to prove that any model of the differential lambda 
calculus and of Taylor expansion~\cite{Manzonetto12} is an instance 
of coherent Taylor expansion~\cite{EhrhardWalch25}. This solves an open problem 
of~\cite{EhrhardWalch25}, and ensures that coherent Taylor expansion is a 
generalization of the previous axiomatization.

\textbf{Related work:} The idea behind $\Tn$ is very similar to higher 
order dual numbers~\cite{Szirmay20} 
and jet bundles that have found recent applications in AD for computing 
higher order derivatives~\cite{Betancourt18,Huot22}.
It is quite likely that a generalization of the content of this 
article to tangent categories~\cite{Cockett14} would provide 
a categorical abstraction of jet bundles. 
%

\textbf{Outline of the paper: } \Cref{sec:cdc} reintroduces 
necessary material on cartesian differential categories and 
higher order derivatives.
\Cref{sec:tangent-bundle} introduces the tangent bundle $\T$,
and reviews how the axioms of cartesian differential categories 
boil down to functoriality and naturality equations on $\T$.
\Cref{sec:iterated-monad} exhibits the canonical monad 
structure of $\T^n$.
\Cref{sec:taylor-expansion} defines $\Tn$, explains why it corresponds to 
a Taylor expansion, and proves its functoriality as a consequence of the
functoriality of $\T^n$.
\Cref{sec:taylor-expansion-monad} proves that
$\Tn$ is a monad, by using the monad structure of $\T^n$.
\Cref{sec:taylor-expansion-sound} proves that any functor satisfying 
the same naturality equations as $\Tn$ is the
Taylor expansion associated to a derivative, 
meaning that the categorical equations of $\Tn$ exactly
capture Taylor expansions.
Finally, \cref{sec:infinitary-sums} considers infinitary 
Taylor expansions in cartesian differential categories 
featuring arbitrary countable sums.
We prove that in that setting, the derivative induces a coherent Taylor 
expansion as defined in \cite{EhrhardWalch25}. Thus, every model of 
\cite{Manzonetto12} is a model of \cite{EhrhardWalch25}.

\section{Cartesian differential categories} \label{sec:cdc}

We recall necessary material on cartesian differential categories,
with the specificity that the homset are not only commutative monoids but 
$\semiring$-semimodules over a commutative semiring $\semiring$.
We mainly recall material from~\cite{Garner21}, with the difference
that we use the word \emph{additive} instead of linear 
in order to stick to the conventions of~\cite{Blute09,EhrhardWalch25}.
A commutative semiring $\semiring$ is a set equipped with 
two commutative monoid structures $(\semiring, 0, +)$ and $(\semiring, 1, \cdot)$ 
such that $a \cdot 0 = 0$ and $a \cdot (b+c) = a \cdot b + a \cdot c$.
For example, the set of natural numbers $\N$ is a commutative semiring, with 
the usual addition and multiplication of natural numbers. For any semiring
$\rig$, there is a canonical semiring morphism 
$\phi : \N \arrow \rig$ given by $\phi(n) = 1_{\rig} + \cdots + 1_{\rig}$.

\begin{definition} \label{def:multiplicative-inverse}
A semiring $\rig$ has multiplicative inverses for integers if for all $n \in \N^*$, 
there exists an element $\frac{1}{n} \in \rig$ such that $\frac{1}{n} \cdot \phi(n) = 1_{\rig}$.
\end{definition}

A $\semiring$-semimodule consists of the data of a commutative monoid $M$
and a multiplicative action $\_ \cdot \_ : \semiring \times M \arrow M$
that respects addition in each variable, and such that 
for all $a, b \in \semiring$ and $x \in M$, 
$(a \cdot b) \cdot x = a \cdot (b \cdot x)$ (this is why the multiplicative
action has the same name as the multiplication in $\semiring$).
A map $f : M \arrow L$ between two $\semiring$-semimodules is 
$\semiring$-additive if $f(x+y) = f(x) + f(y)$ and 
$f(a \cdot x) = a \cdot f(x)$ (the case 
$a = 0$ ensures that $f(0) = 0$).

\begin{notation}
For any cartesian category $\cat$, we write the categorical product 
$X_0 \times X_1$, the pairing of $f \in \cat(X, Y_0)$ with $g \in \cat(X, Y_1)$
as $\prodPair{f}{g} \in \cat(X, Y_0 \times Y_1)$, and the projections as
$\proj_0 \in \cat(X_0 \times X_1, X_0)$ and $\proj_1 \in \cat(X_0 \times X_1, X_1)$.
We write $X^{\times n}$ and $f^{\times n}$ the cartesian product of $X$ and $f$ 
with itself taken $n$-times.
\end{notation}

\begin{definition}
A left $\semiring$-additive category 
is a category $\cat$ such that each hom-set $\cat(X, Y)$
has a $\semiring$-semimodule structure and such that for all object $Z$ of $\cat$ and
$f \in \cat(X, Y)$, the map 
$\_ \comp f : \cat(Y, Z) \arrow \cat(X, Z)$ is $\semiring$-additive.

A morphism $h \in \cat(Y, Z)$ is $\semiring$-additive if for all object $X$,
the map $h \comp \_ : \cat(X, Y) \arrow \cat(X, Z)$ is $\semiring$-additive. 
A $\semiring$-additive category is a left $\semiring$-additive 
category in which all the morphisms are $\semiring$-additive.  

A cartesian left $\semiring$-additive category is a left $\semiring$-additive 
category that is cartesian and such that the projections of the cartesian product 
are $\semiring$-additive. This last condition means that the $\semiring$-module structure is 
compatible with the tupling, in the sense that
$\prodPair{f_0}{g_0} + \prodPair{f_1}{g_1} 
= \prodPair{f_0+f_1}{g_0 + g_1}$ 
and for all $r \in \semiring$, 
$r \cdot \prodPair{f}{g} = \prodPair{r \cdot f}{r \cdot g}$. 
\end{definition}
%

Given any cartesian left $\semiring$-additive category,
it is straightforward to check that there exists 
a $\semiring$-additive category $\catAdd$ whose objects are 
the objects of $\cat$, and whose morphisms are the 
$\semiring$-additive morphisms of $\cat$.
The composition of $f \in \catAdd(X, Y)$ with $g \in \catAdd(Y, Z)$
is $g \comp f$, but we write this composition 
$g \compl f$ to stress the "linear" nature of the morphisms.
By $\semiring$-additivity of the projections, $\catAdd$ is
also cartesian. 


\begin{definition}
A map $f \in \cat(X_0 \times \cdots \times X_n, Y)$ is $k$-additive 
in an argument $i$ if the map 
$g \mapsto f \comp (\prod_{l=0}^{i-1} \id_{X_l} \times g \times 
\prod_{l=i+1}^n \id_{X_l})$
is $\semiring$-additive.
\end{definition}

Let $\Sn$ be a functor on $\catAdd$ defined as 
$\Sn X = X \times X^{\times n}$ and $\Sn f = f \times f^{\times n}$.
In particular, let $\S = \Sn[1]$, so that $\S X = X \times X$ and 
$\S f = f \times f$.
For all $i \in \interval{0}{n}$, we write $\Sproj_i 
= \proj_i : \Sn \naturalTrans \idfun$
($\Sproj_i$ and $\proj_i$ are the same, but we use different 
names to stress their respective role).
Intuitively, an element of $\Sn X$ is a vector 
$(x, u_1, \ldots, u_n)$, where $x$ is the base point, and $u_i$ is an order 
$i$ variation around $x$. 
The following definition corresponds to
Def 2.4 of~\cite{Garner21}, with a slight variation on 
\ref{def:cdc-linear} and \ref{def:cdc-schwarz} that are equivalent to 
items (vi) and (vii) of~\cite{Garner21} thanks to Prop. 4.2
of~\cite{Cockett14}. 
\begin{definition} \label{def:cdc}
A derivative in a cartesian left $\semiring$-additive category $\cat$ 
is a family of operators
$\d : \cat(X, Y) \arrow \cat (\S X, Y)$ for all objects $X, Y$ such that \begin{itemize}
\item[\ref{def:cdc-projections}] for any $\proj_i \in \cat(X_0 \times \cdots \times X_n, X_i)$, 
$\d \proj_i = \proj_i \compl \Sproj_1$; \labeltext{$\d$-proj}{def:cdc-projections}
\item[\ref{def:cdc-additive}] $\d$ is $\semiring$-additive; \labeltext{$\d$-sum}{def:cdc-additive}
\item[\ref{def:cdc-chain}] $\d \id = \Sproj_1$ and $\d(g \comp f) = \d g \comp \prodPair{f \comp \Sproj_0}{\d f}$
\labeltext{$\d$-chain}{def:cdc-chain};
\item[\ref{def:cdc-left-additive}] $\d f$ is $\semiring$-additive in its second argument 
\labeltext{$\d$-add}{def:cdc-left-additive};
\item[\ref{def:cdc-linear}] $\d \d f \comp \prodtuple{x, 0, 0, u} 
= \d f \comp \prodtuple{x, u}$; \labeltext{$\d$-lin}{def:cdc-linear}
\item[\ref{def:cdc-schwarz}] $\d \d f \comp \prodtuple{x, u, v, w} =
\d \d f \comp \prodtuple{x, v, u, w}$. \labeltext{$\d$-sym}{def:cdc-schwarz}
\end{itemize}
A cartesian $\rig$-differential category is a cartesian 
left $\semiring$-additive category equipped with a derivative. 
\end{definition}
We refer to sec. 2.2 of~\cite{Garner21} for a more involved description 
of those axioms.
As explained in~\cite{Garner21}, any monoid is a $\N$-semimodule, so 
the cartesian differential 
categories of~\cite{Blute09} based
on left additive categories where the hom-sets are only assumed to be monoids are
instances of this definition with $\semiring = \N$. 
The most fundamental example of cartesian differential category is the 
category $\SMOOTH$ whose objects are the Euclidean vector spaces 
$\R^n$ and whose maps are the smooth functions. The derivative 
$\d$ in $\SMOOTH$ is the usual derivative of functions, 
$\d f(x, u) = \derive{f} \cdot u$. 

\begin{proposition}[Lemma 2.6 of \cite{Lemay18}] \label{prop:D-pairing}
    For any operator $\d$ that satisfies 
    \ref{def:cdc-projections} and \ref{def:cdc-chain}, 
    $\d \prodPair{f}{g} = \prodPair{\d f}{\d g}$.
\end{proposition}


\begin{definition}
A map $f \in \cat(X, Y)$ is $\d$-linear if $\d f = f \comp \Sproj_1$. 
\end{definition}
In $\SMOOTH$, this definition means that $\derive{f} \cdot u = f(u)$
so the $\d$-linear maps are the linear maps between Euclidian spaces.
There exists a category $\catLin$ with the same objects as
$\cat$, whose morphisms are the $\d$-linear morphisms of $\cat$, 
and whose composition coincides with the composition in $\cat$.
It follows from~\cref{prop:D-pairing} that $\catLin$ is cartesian, with the same projections
and tupling as in $\cat$.
By \ref{def:cdc-left-additive}, any $\d$-linear morphism is also 
$\semiring$-additive, but the converse is not true in general, see~\cite{Blute09}.
Thus, $\catLin$ is a subcategory of $\catAdd$. We write the composition 
of $f$ with $g$ as $g \compl f$, as in $\catAdd$.
$\SMOOTH$ is a particular case in which $\LIN{\SMOOTH} = \ADD{\SMOOTH}$.

\begin{definition}
Given $f \in \cat(X_0 \times \cdots \times X_n, Y)$ and $i \in \interval{0}{n}$,
the $i$-th partial derivative of $f$ is the morphism
$\d_i f \in \cat(X_0 \times \cdots \times X_n \times X_i, Y)$ defined as
$\d_i f = \d f \comp \prodtuple{\proj_0, \ldots, \proj_n, 0^{\times i}, \proj_{n+1}, 0^{\times n-i}}$.

A morphism $f \in \cat(X_0 \times \cdots, \times X_n, Y)$ is $\d$-linear in its 
$i$-th argument if $\d_i f = f \comp \prodtuple{\proj_0, \ldots, 
\proj_{i-1}, \proj_{n+1}, \proj_{i+1}, \ldots, \proj_n}$.
\end{definition}

In $\SMOOTH$, this partial derivative $\d_i$ coincides with the usual 
partial derivative $\partial_i$, this is an immediate consequence of the fact 
that for all smooth map $f : X_0 \times \cdots \times X_n \arrow Y$, 
$\derive{f} \cdot (u_0, \ldots, u_n) = \sum_{i=0}^n \partial_i f(x) \cdot u_i$.

\begin{definition}
The $n$-th derivative of $f \in \cat(X, Y)$
is the map $\hod{f}{n} = (\d_1)^n f \in \cat(\Sn X, Y)$.
\end{definition}

Let us detail the relationship between the 
$n$-th derivative $\hod{f}{n} \in \cat(\Sn X, X)$
and the $n$-th \emph{total} derivative 
$\d^n f \in \cat(\S^n X, Y)$.
For any set $I$, let $\upart{I}$ be the set of
partitions of $I$, and $\opart I$ the set of ordered partitions of 
$I$. More precisely, an element of $\upart I$ is a set 
$\{I_1, \ldots, I_k\}$ and an element of 
$\opart I$ is a vector $(I_1, \ldots, I_n)$ such that
$I_1, \ldots, I_n$ is a partition of $I$.
This notation is based on (unordered) Bell numbers, 
that give the numbers of (unordered) partitions of a set.

\begin{notation} \label{notation:set-projection}
    Let $\intsegment{n} = \{1, \ldots, n\}$. 
    There is a bijection between words $w \in \{0, 1\}^*$ of length $\length{w} = n$ 
    and subsets of $I_w \subseteq \intsegment{n}$, given by 
    $I_w = \{i \in \intsegment{n} \st w_i = 1\}$.
    The weight $\weight{w}$ is the number 
    of $1$ in $w$ (and not its length as it is often standard), so that 
    the weight $\weight w$ matches the size $\cardinal{I_w}$ of its corresponding set. 
    We will often consider sets and words up to this bijection.
    For example, $\upart{w}$ is a set whose elements are the
    sets $\{w^1, \ldots, w^k\}$ of words of weight at least one 
    and such that 
    for all $i \in \intsegment{\length{w}}$, 
    $\sum_{j=1}^k w_i^k = w_i$ 
    For any word $w = w_1 \cdots w_n$ of length $n$, we write
    $\Sproj_w = \Sproj_{w_1} \compl \cdots \compl \Sproj_{w_n}$.
    Let $\zeroword$ and $\oneword$ be the words 
    $0 \cdots 0$ and $1 \cdots 1$ 
    of length $n$ (the context in which this notation is used 
    will make clear that this corresponds to a word).
\end{notation}

Intuitively, $S^n X$ can be seen as the set of 
polynomials of degree $1$ over $X$, with 
$n$ formal indeterminates $\formalvar_1, \ldots, \formalvar_n$ 
such that for all $i$, $\formalvar_i^2 = 0$. The projection 
$\Sproj_{w_1 \cdots w_n}(a)$ provides the coefficient in $a$ of the 
monomial $\formalvar_1^{w_1} \cdots \formalvar_n^{w_n}$.

\begin{proposition}[Lemma 3.2 of~\cite{Garner21}] \label{prop:higher-order-derivative}
For all $f \in \cat(X, Y)$: \begin{enumerate}
    \item $\hod{f}{n}$ is symmetric and 
    $\d$-linear in its last $n$ variables;\label{prop:higher-order-derivative-1}
    \item $\d^n f = \sum_{\{w^1, \ldots, w^k\} \in \upart{\intsegment n}} \hod{f}{k}
    \comp \prodtuplenosize{\Sproj_{\zeroword}, \Sproj_{w^1}, \ldots, \Sproj_{w^k}}$,
    this is well-defined by symmetry of the $\hod{f}{k}$; 
    \label{prop:higher-order-derivative-2}
    \item if $\rig$ has multiplicative inverses for integers, then
    $\d^n f = \sum_{(w^1, \ldots, w^k) \in \opart{\intsegment n}} 
    \frac{1}{\factorial{k}}\hod{f}{k}
    \comp \prodtuplenosize{\Sproj_{\zeroword}, \Sproj_{w^1}, \ldots, \Sproj_{w^k}}$.
    \label{prop:higher-order-derivative-3}
\end{enumerate}
\end{proposition}

\begin{proof}
Proofs for \cref{prop:higher-order-derivative-1,prop:higher-order-derivative-2}
are in~\cite{Garner21}. \Cref{prop:higher-order-derivative-3} is 
a direct consequence of \cref{prop:higher-order-derivative-2}.
\end{proof}

For example, \cref{prop:higher-order-derivative-2,prop:higher-order-derivative-3} 
for $n=2$ implies that
\begin{equation} \label{eq:d-two}
    \d \d f \comp \prodtuple{x, u, v, w} 
= \hod f 2 \comp \prodtuple{x, u, v} + \d f \comp \prodPair{x}{w} 
\end{equation}
which can be understood in $\SMOOTH$ by the fact that the derivative of the 
map $g : (x, u) \mapsto \derivenoarg{f}(x) \cdot u$
is equal to 
\begin{align*} 
    g'(x, u) \cdot (v, w) &= \partial_0 g(x, u) \cdot v + 
\partial_1 g(x, u) \cdot w \\
&= \hod f 2 (x) \cdot (u, v) + \derivenoarg f(x) \cdot w
\end{align*}
by definition of $\hod f 2(x)$
and linearity of $u \mapsto \derivenoarg{f}(x) \cdot u$.

\begin{example} \label{ex:wrel}
Cartesian differential categories notably include the models of the differential 
$\lambda$-calculus \cite{Bucciarelli10} and its 
syntactical Taylor expansion \cite{Manzonetto12}. 
One example of such model is 
the weighted relational model $\klWREL$~\cite{Laird13}.
Let $\rig$ be a \emph{complete} semiring (a semiring with arbitrary
countable sums), for example
$\rig = \overline{\R_{\geq 0}} = \R_{\geq 0} \cup \infty$.
A finite multiset over $A$ is a function $m : A \arrow \N$ such that 
$\{a \in A \st m(a) \neq 0 \}$ is finite.
The objects of the category $\klWREL$ (it arises as the coKleisli category of 
the $\oc$ comonad of a model of Linear Logic) 
are the (countable) sets,
and the morphism from $A$ to $B$ are the matrices
$\rig^{\mfin(A) \times B}$, where $\mfin(A)$ is the set of all finite 
multisets over $A$. 
%

The category $\WREL$ is cartesian, the product 
of $A$ and $B$ is the disjoint union: 
$A \uplus B = \{(0, a) \st a \in A\} \cup \{(1, b) \st b \in B\}$.
In particular, we have $\mfin{(A \uplus B)} = \mfin(A) \times \mfin(B)$.
The category $\WREL$ is a cartesian differential category.
For all $f \in \klWREL(A, B)$, the higher order derivative 
$\hod f n \in \klWREL(\uplus_{i=0}^n A, B)$ can be described 
by a matrix $\hod f n \in \rig^{\mfin(A) \times \mfin(A)^n \times B}$ 
such that
\[ f_{m, [a_1], \ldots, [a_n], b}
= \frac{\factorial{(m + [a_1, \ldots, a_n])}}{\factorial{m}} 
    f_{m+[a_1, \ldots, a_n], b} \]
and all other coefficients are set to $0$. 
In the above, $[a]$ is the multiset with one element, and 
$\factorial{m} = \prod_{a \in A} \factorial{m(a)}$.
A morphism $s \in \klWREL(\uplus_{i=1}^n A_i, B)$
is $\d$-linear in a coordinate $i$ if and only if 
$s_{p_1, \ldots, p_n} \neq 0 \imply p_i = [a_i]$ 
for some $a_i \in A_i$, so $\d$-linearity coincides with the linearity 
of linear logic.
\end{example}

\section{The tangent bundle construction} \label{sec:tangent-bundle}

We now introduce the tangent bundle functor $\T$ and 
prove that the equations of derivatives
are equivalent to naturality equations that turns $\T$ 
into a monad.
The functorialization of differentiation can be found in two 
(orthogonal) generalizations of cartesian differential categories: 
tangent categories~\cite{Rosicky84,Cockett14} and cartesian coherent 
differential categories~\cite{Walch23} that differ by their action on objects.
%
%
Both of those theories end up identical when restricted 
to cartesian differential categories, with action on object 
$\T X = X \times X$.
%

There is a bijection between 
operators $\d : \cat(X, Y) \arrow \cat(\S X, Y)$ (not assumed 
to be derivatives) and operators
$\T : \cat(X, Y) \arrow \cat(\S X, \S Y)$ such that 
$\Sproj_0 \comp \T f = f \comp \Sproj_0$.
\begin{equation} \label{eq:tangent-bundle-cdc}
    \T f = \prodPair{f \comp \Sproj_0}{\d f} \qquad
    \d f = \Sproj_1 \comp \T f
\end{equation}

\begin{proposition} \label{prop:tangent-bundle-functor}
The following assertions are equivalent. \begin{enumerate}
    \item $\d$ satisfies \ref{def:cdc-chain}: 
    $\d \id = \Sproj_1$ and $\d (g \comp f) = \d g \comp \T f$.
    \item $\T$ is a functor with action on objects $\T X = \S X$.
\end{enumerate}
\end{proposition}
\begin{proof} Observe that 
$\T (g \comp f) = \prodPair{g \comp f \comp \Sproj_0}{\d (g \comp f)}$ 
and $\T g \comp \T f = \prodPair{g \comp \Sproj_0}{\d g} \comp 
\prodPair{f \comp \Sproj_0}{\d f} 
= \prodPair{g \comp f \comp \Sproj_0}{\d g \comp \T f}$
so $\T (g \comp f) = \T g \comp \T f$ if and only if 
$\d (g \comp f) = \d g \comp \T f$. 
Similarly, $\T \id = \id$ if and only if 
$\d \id = \Sproj_1$.
\end{proof}

This functorial point of view on the derivative is 
used in~\cite{Lemay18} to provide an alternative to the 
heavy combinatorics of the Fa\'a di Bruno formula.
%
Indeed, by definition $\d^n f = \Sproj_{\oneword} \comp \T^n f$,
so by functoriality of $\T^n$ it immediately follows that 
\begin{equation} \label{eq:faa-di-bruno-functor}
\d^n (g \comp f) = \d^n g \comp \T^n f .
\end{equation}
The simplicity of this equation sharply contrasts with 
the combinatorial Fa\'a di Bruno formula 
(cor. 3.2.3 of~\cite{Cockett11})
\begin{equation} \label{eq:faa-di-bruno}
\hod{(g \comp f)}{n} = \! \! \! \! \! \! \! \! \! \! \! \! \! 
\sum_{\{I_1, \ldots, I_k\} \in \upart {\intsegment{n}}} \! \! \!
\! \! \! \! \! \! \hod{g}{k} \comp \prodtuple{\hod{f}{\emptyset}, 
\hod{f}{I_1}, \ \ldots \ , \hod{f}{I_k}} 
\end{equation}
where for any set $I = \{i_1, \ldots, i_m\} \subseteq \intsegment{n}$, 
$\hod{f}{I} = 
\hod{f}{m} \comp \prodtuple{\Sproj_0, \Sproj_{i_1}, \ldots, \Sproj_{i_m}}$
(by symmetry of $n$-th derivative, this value does not depend
on the choice of ordering of $I$).

\begin{proposition} \label{prop:T-linear}
    A morphism $f \in \cat(X, Y)$ is $\d$-linear 
iff $\T f = \S f$. Thus, $\T$ coincides with $\S$
when restricted to $\catLin$.
\end{proposition}

\version{}{
\begin{proof}
If $f$ is $\d$ linear then 
$\T f = \prodPair{f \comp \Sproj_0}{\d f} 
= \prodPair{f \comp \Sproj_0}{f \comp \Sproj_1} 
= \S f$. Conversely, if $\T f = \S f$ then 
$\d f = \Sproj_1 \comp \T f = \Sproj_1 \comp \S f = f \comp \Sproj_1$. 
\end{proof}}

We now define natural transformations in $\catAdd$ and $\catLin$.
\begin{equation} 
    \begin{split}
    \Ssum = \Sproj_0 + \Sproj_1 &: \S \naturalTrans \idfun \\
    \text{for all } r \in \rig, \homothety r = r \cdot \id &: \idfun \naturalTrans \idfun
    \end{split}
\end{equation}
that consists of a sum operation and the homotheties.
These are natural transformations by $\semiring$-additivity 
of the morphisms in $\catAdd$.
We also define the following natural transformations
\begin{equation} \label{eq:cdc-natural-equations}
\begin{split}
    \Sinjz = \prodPair{\id}{0} &: \idfun \naturalTrans \S \\
    \SmonadSum = \prodPair{\Sproj_0 \compl \Sproj_0}{\Sproj_0 \compl \Sproj_1
    + \Sproj_1 \compl \Sproj_0} &: \S^2 \naturalTrans \S \\
    \Sswap = \prodPair{\prodPair{\Sproj_0 \compl \Sproj_0}{\Sproj_0 \compl \Sproj_1}}
    {\prodPair{\Sproj_1 \compl \Sproj_0}{\Sproj_1 \compl \Sproj_1}} 
    &: \S^2 \naturalTrans \S^2 \\
    \Slift = \prodPair{\prodPair{\Sproj_0}{0}}{\prodPair{0}{\Sproj_1}} 
    &: \S \naturalTrans \S^2 \\
    \text{for all } r \in \semiring, \Sscale r = \prodPair{\Sproj_0}{r \cdot \Sproj_1} &: \S \naturalTrans 
    \S
\end{split}
\end{equation}
The naturality follows from the fact that the pairing 
$\prodPair{\alpha}{\beta}$ of two natural transformations 
$\alpha, \beta : F \naturalTrans G$ (where $F, G$ are functors whose codomain 
is $\catAdd$ or $\catLin$) is a natural transformation 
$F \naturalTrans \S G$, and that their sum is a natural transformation 
$F \naturalTrans G$. 
Observe that those natural transformations are natural with regard to 
$\S$, and not $\T$. In fact, the naturality 
with regard to $\T$ of the families of~\cref{eq:cdc-natural-equations} corresponds \emph{exactly}
to the axioms of cartesian differential categories, as discovered in~\cite{Walch23}.
The forward implication of this observation was first discovered in 
sec. 2.4 of~\cite{Cockett14}
in the proof that every cartesian differential category is a tangent category.
\begin{theorem} \label{thm:tangent-bundle-natural} 
    Let $\d : \cat(X, Y) \arrow \cat(\S X, Y)$ be an operator
    that satisfies \ref{def:cdc-projections} and \ref{def:cdc-chain}.
    Then: 
    \begin{enumerate}
        \item $\d$ satisfies \ref{def:cdc-additive} iff $\T \Ssum = \S \Ssum$ and 
        $\T (\homothety r) = \S (\homothety r)$ for all $r \in \rig$ (that is, $\homothety r$ and $\Ssum$ are 
        $\d$-linear);
        \item $\d$ satisfies \ref{def:cdc-left-additive} iff 
        $\Sinjz : \idfun \naturalTrans \T$, $\SmonadSum : \T^2 \naturalTrans \T$
        and $\Sscale{r} : \T \naturalTrans \T$ (for all 
        $r \in \semiring$) are natural in $\cat$\label{thm:tangent-bundle-natural-additive};
        \item Assuming that $\Sinjz : \id \naturalTrans \T$ is natural in $\cat$, 
        $\d$ satisfies \ref{def:cdc-linear} iff $\Slift : \T \naturalTrans \T^2$
        is natural in $\cat$;
        \item $\d$ satisfies \ref{def:cdc-schwarz} iff $\Sswap : \T^2 \naturalTrans \T^2$ is natural
        in $\cat$.
    \end{enumerate}
    Thus, there is a bijection between derivative operators $\d$ and functors 
$\T$ such that $\T X = \S X$, $\T \proj_i = \S \proj_i$, 
$\T \Ssum = \S \Ssum$, $\T \homothety r = \S \homothety r$ and such that 
$\Sproj_0$, $\Sinjz$, $\SmonadSum$, $\Sscale{r}$ (for all $r \in \semiring$), 
$\Slift$ and $\Sswap$ are natural in $\cat$ with 
regard to $\T$.
\end{theorem}
This theorem and its proof, 
\version{see sec. A of the long version 
of this paper~\cite{Walch25-vlong}}{see \cref{appendix:tangent-bundle}},
are a slight adaptation of thm. 7 of~\cite{Walch23}. The difference is
that~\cite{Walch23} deals with partial 
sums but does not feature a $\semiring$-module structure. 
The proof idea is similar to that of
\cref{prop:tangent-bundle-functor}: 
the equations of \cref{def:cdc} correspond 
to the rightmost projections of the equations on $\T$, 
while the other projections always hold by definition of $\T$.

\begin{proposition} \label{prop:tangent-bundle-monad}
$(\S, \Sinjz, \SmonadSum)$ is a monad on $\catAdd$ and 
$\catLin$, meaning that the following diagram 
commute.
\[ 
\begin{tikzcd}[row sep = small]
	\S & {\S^2} & \S \\
	& \S
	\arrow["{\Sinjz \S}", from=1-1, to=1-2]
	\arrow[Rightarrow, no head, from=1-1, to=2-2]
	\arrow["\SmonadSum", from=1-2, to=2-2]
	\arrow["{\S \Sinjz}"', from=1-3, to=1-2]
	\arrow[Rightarrow, no head, from=1-3, to=2-2]
\end{tikzcd} \quad 
\begin{tikzcd}[row sep = small]
	{\S^3} & {\S^2} \\
	{\S^2} & \S
	\arrow["{\S \SmonadSum}", from=1-1, to=1-2]
	\arrow["{\SmonadSum \S}"', from=1-1, to=2-1]
	\arrow["\SmonadSum", from=1-2, to=2-2]
	\arrow["\SmonadSum"', from=2-1, to=2-2]
\end{tikzcd} \]
$(\T, \Sinjz, \SmonadSum)$ is a monad on $\cat$ (same diagrams, replacing
$\S$ by $\T$).
\end{proposition}
\begin{proof} 
Let us prove the rightmost diagram.
\begin{align*} 
    \Sproj_i \compl \SmonadSum_X \compl \SmonadSum_{\S X} &= 
\sum_{i_1 + i_2 = i} \Sproj_{i_1} \compl \Sproj_{i_2} \compl \SmonadSum_{\S X}
= \sum_{j_1 + j_2 + j_3 = i} \Sproj_{j_1} \compl \Sproj_{j_2} \compl \Sproj_{j_3} \\
\Sproj_i \compl \SmonadSum_X \compl \S \SmonadSum_{X} &= 
\sum_{i_1 + i_2 = i} \Sproj_{i_1} \compl \SmonadSum_{X} \compl \Sproj_{i_2} 
= \sum_{j_1 + j_2 + j_3 = i} \Sproj_{j_1} \compl \Sproj_{j_2} \compl \Sproj_{j_3} 
\end{align*}
so $\SmonadSum \compl \SmonadSum \S = \SmonadSum \compl \S \SmonadSum$.
The proof of the left diagram is very similar, and we conclude that
$(\S, \Sinjz, \SmonadSum)$ is a monad. 
We also conclude that $(\T, \Sinjz, \SmonadSum)$ is a monad on $\cat$
by naturality of $\Sinjz : \idfun \naturalTrans \T$ and 
$\SmonadSum : \T^2 \naturalTrans \T$ (see \cref{thm:tangent-bundle-natural}) 
and because $\T$ coincides with $\S$ on $\catLin$ by~\cref{prop:T-linear}.
\end{proof}
The \emph{Kleisli category} $\kleisliT$ of the monad $\T$ is the category
whose objects are the objects of $\cat$, and such that 
$\kleisliT(X, Y) = \cat(X, \T Y)$. The composition of $f \in \kleisliT(X, Y)$
with $g \in \kleisliT(Y, Z)$ is given by $\SmonadSum \comp \T g \comp f$.
A morphism $f \in \kleisliT(X, Y)$ corresponds to a polynomial 
$f_0 + f_1 \formalvar$ with coefficients in $\cat(X, Y)$ and a formal indeterminate 
$\formalvar$ such that $\formalvar^2 = 0$. The composition of $f$ with 
$g = g_0 + g_1 \formalvar$ is equal to 
$g_0 \comp f_0 + (g_1 \comp f_0 + \d{g_0} \comp \prodPair{f_0}{f_1}) \formalvar$.
This monad captures the concept of dual numbers widely used 
in automated differentiation~\cite{Baydin17}.

\begin{notation}
    Define 
    $\kronecker i j = 1$ if $i=j$, and $\kronecker i j = 0$ otherwise.
\end{notation}

\begin{proposition} \label{prop:tangent-bundle-distributive}
    The natural transformation $\Sswap$ is a distributive law~\cite{Beck69} of 
the monad $\S$ over itself. This means that the following diagrams commute 
in $\catAdd$ and $\catLin$.
      \[ \begin{tikzcd}[row sep = small]
        \S \\
        {\S \S} & {\S \S}
        \arrow["{\Sinjz \S}"', from=1-1, to=2-1]
        \arrow["\Sswap"', from=2-1, to=2-2]
        \arrow["{\S \Sinjz}", from=1-1, to=2-2]
      \end{tikzcd} \quad 
      \begin{tikzcd}[row sep = small]
        {\S^2 \S} & {\S \S \S} & {\S \S^2} \\
        {\S \S} && {\S \S}
        \arrow["{\S \Sswap}", from=1-1, to=1-2]
        \arrow["{\Sswap \S}", from=1-2, to=1-3]
        \arrow["{\SmonadSum \S}"', from=1-1, to=2-1]
        \arrow["\Sswap"', from=2-1, to=2-3]
        \arrow["{\S \SmonadSum}", from=1-3, to=2-3]
      \end{tikzcd} \]
      \vspace{-0.5em}
      \[ \begin{tikzcd}[row sep = small]
        \S \\
        {\S \S} & {\S \S}
        \arrow["{\Sinjz \S}"', from=1-1, to=2-1]
        \arrow["\Sswap", from=2-2, to=2-1]
        \arrow["{\S \Sinjz}", from=1-1, to=2-2]
      \end{tikzcd} \quad 
      \begin{tikzcd}[row sep = small]
        {\S^2 \S} & {\S \S \S} & {\S \S^2} \\
        {\S \S} && {\S \S}
        \arrow["{\S \Sswap}"', from=1-2, to=1-1]
        \arrow["{\Sswap \S}"', from=1-3, to=1-2]
        \arrow["{\SmonadSum \S}"', from=1-1, to=2-1]
        \arrow["\Sswap", from=2-3, to=2-1]
        \arrow["{\S \SmonadSum}", from=1-3, to=2-3]
      \end{tikzcd} \]
      Similarly, $\Sswap$ is a distributive law of $\T$ over itself 
      in $\cat$.
\end{proposition}

\begin{proof} We prove that $\Sswap_X \compl (\Sinjz \S)
    = \S \Sinjz$ as follows.
    \begin{alignat*}{3} \Sproj_i \compl \Sproj_j \compl \Sswap_X \compl \Sinjz_{\S X}
    &=  \Sproj_j \compl \Sproj_i \compl \Sinjz 
    &&= \kronecker{0}{i} \Sproj_j \\
    \Sproj_i \compl \Sproj_j \compl \S \Sinjz_{X} 
    &= \Sproj_i \compl \Sinjz_{X}  \compl \Sproj_j 
    &&= \kronecker{0}{i} \Sproj_j
    \end{alignat*}
    We prove that $\Sswap \compl (\SmonadSum \S) = 
    (\S \SmonadSum) \compl (\Sswap \S) \compl \S \Sswap$ as follows.
    \[
        \Sproj_i \compl \Sproj_j \compl \Sswap_X \compl \SmonadSum_{\S X}
    = \Sproj_j \compl \Sproj_i \compl \SmonadSum_{\S X} 
    = \Sproj_j \compl \smashoperator{\sum_{i_1 + i_2 = i}}  \Sproj_{i_1} \compl \Sproj_{i_2} 
    =  \smashoperator{\sum_{i_1 + i_2 = i}} \Sproj_{j} \compl \Sproj_{i_1} \compl \Sproj_{i_2} \] 
    \begin{multline*}
        \Sproj_i \compl \Sproj_j \compl \S \SmonadSum_X \compl 
    \Sswap_{\S X} \compl \S \Sswap_X = 
    \Sproj_i \compl \SmonadSum_X \compl \Sproj_j \compl \Sswap_{\S X} \compl \S \Sswap_X \\ 
    = \sum_{i_1 + i_2 = i} \Sproj_{i_1} \compl \Sproj_{i_2} \compl \Sproj_j \compl 
    \Sswap_{\S X} \compl \S \Sswap_X 
    = \sum_{i_1 + i_2 = i} \Sproj_j \compl \Sproj_{i_1} \compl \Sproj_{i_2} 
    \end{multline*}
    The two other diagrams hold upon observing that $\Sswap$ is 
    an involution: $\Sswap \compl \Sswap = \id$.
    We conclude that $\Sswap$ is also a distributive law of $\T$ over itself, 
    by naturality of $c : \T \T \naturalTrans \T \T$ (\cref{thm:tangent-bundle-natural})
    and because $\T$ coincides with $\S$ on $\catLin$.
\end{proof}

\section{The iterated tangent bundle monad} \label{sec:iterated-monad}

This section presents the following new result: $\T^n$ has a canonical monad 
structure that arises from the distributive law $\Sswap$ of $\T$ over itself.
The naturality equations of this structure 
provide a higher order variant of the axioms of cartesian differential categories, 
in the same way that the functoriality of $\T^n$ provides a higher order 
variant of the chain rule~\cref{eq:faa-di-bruno-functor}.
Let us provide some intuition on $\T^n f$. 
An element
$a \in \S^n X$ intuitively corresponds to polynomials over $X$ in $n$ formal indeterminates
$\formalvar_1, \ldots, \formalvar_n$ such that $\formalvar_i^2 = 0$, 
\[ a(\formalvar_1, \ldots, \formalvar_n) = \smashoperator{\sum_{w_1, \ldots, w_n \in \{0,1\}}} 
a_{w_1, \ldots, w_n} \formalvar_1^{w_1} \cdots \formalvar_n^{w_n} \]
where $a_{w_1, \ldots, w_n} = \Sproj_{w_1, \ldots, w_n}(a)$. 
Then, $\T^n f \in \cat(\S^n A, \S^n B)$ 
can be seen as a map between polynomials defined as follows. 
Let $b = a - a_{\zeroword}$, so that by Taylor expansion
\[f(a) \simeq f(a_{\zeroword}) + \sum_{k = 1}^n 
\frac{1}{\factorial{k}} \hod{f}{k}(a_{\zeroword})(b, \ldots, b) + 
O(b^{n+1}). \]
But then, $b^{n+1} = 0$ because $\formalvar_i^2 = 0$ for all $i$. 
So all the higher order term of the Taylor expansion vanish.
This expression can then be developed by linearity and 
multilinearity of the derivatives, and simplified thanks to 
the equalities $\formalvar_i^2 = 0$. This 
yields a polynomial over $\formalvar_1, \ldots, \formalvar_n$. This polynomial 
precisely corresponds to the polynomial 
$\T^n f (a)$. 

    For any functors $F, F': \cat \arrow \catbis$ and $G, G' : \catbis \arrow \catter$ and 
    natural transformations $\alpha : F \naturalTrans F'$, $\beta : G \naturalTrans G'$, 
    we write the horizontal composition $\alpha \hcomp \beta : G F \naturalTrans G' F'$.
    \[ (\alpha \hcomp \beta)_X = 
    \beta_{F' X} \comp G \alpha_X = F' \alpha_X \comp \beta_{F X}\] 
As it is standard~\cite{Beck69}, the existence of the distributive law $\Sswap$ 
of $\S$ over itself
induces a monad structure on $\S^2$ given by 
\[ \begin{tikzcd}[row sep = abysmal]
    \Sinjzn[2] = \idfun & & \S \S \\
    \SmonadSumn[2] = \S \S \S \S & {\S \S \S \S} & {\S \S.}
    \arrow["{\Sinjz \hcomp \Sinjz}", from=1-1, to = 1-3]
	\arrow["{\S \Sswap \S}", from=2-1, to=2-2]
	\arrow["{\SmonadSum \hcomp \SmonadSum}", from=2-2, to=2-3]
\end{tikzcd} \]
We describe a similar structure for any $n \in \N$.
The natural transformation $\Sswap$ satisfies the \emph{Yang-Baxter equation}.
\[ 
\begin{tikzcd}[row sep = none]
	& {\S \S \S} & {\S \S \S} \\
	{\S \S \S} &&& {\S \S \S} \\
	& {\S \S \S} & {\S \S \S}
	\arrow["{\S \Sswap}", from=1-2, to=1-3]
	\arrow["{\Sswap \S}", from=1-3, to=2-4]
	\arrow["{\Sswap \S}", from=2-1, to=1-2]
	\arrow["{\S \Sswap}"', from=2-1, to=3-2]
	\arrow["{\Sswap \S}"', from=3-2, to=3-3]
	\arrow["{\S \Sswap}"', from=3-3, to=2-4]
\end{tikzcd} \] 
By thm 2.1 of~\cite{Cheng11}, this induces a canonical monad structure on 
$\S^n$ for any $n \in \N$, and a canonical distributive 
law of $\S^n$ over $\S^m$.
Let us provide an explicit description of this structure.
Define natural transformations 
$\Sinjzn : \idfun \naturalTrans \S^n$ and $\SmonadSumn : \S^{2n} \naturalTrans \S^n$,
$\Sswapnm : \S^n \S^m \naturalTrans \S^m \S^n$ and $\Sscalen r : \S^n \naturalTrans \S^n$
in $\catAdd$ and $\catLin$ by
    \begin{equation} \label{eq:iterated-monad}
        \begin{split} 
        \Sproj_w \compl \Sinjzn = \kronecker{w}{\zeroword} \id \qquad
        & \Sproj_w \compl \SmonadSumn = \hspace{-1em} \sum_{(w^1, w^2) \in \opart w} \hspace{-1em}
        \Sproj_{w^1} \compl \Sproj_{w^2} \\
        \Sproj_w \compl \Sscalen r = r^{\weight{w}} \Sproj_w \qquad
        & \Sproj_u  \compl \Sproj_v \Sswapnm =  \Sproj_v \compl \Sproj_u 
        \end{split}
    \end{equation}

\begin{proposition} \label{prop:iterated-monad-S}
    For all $n \in \N$, $\S^n$ is a monad 
    with unit $\Sinjzn$ and multiplication 
    $\SmonadSumn : \S^n \S^n \naturalTrans \S^n$. 
    Furthermore, $\Sswapnm$ is a distributive law of $\S^n$ over 
    $\S^m$.
\end{proposition}

\begin{proof}
    The proofs of the monad diagrams is the same as the proof of 
    \cref{prop:tangent-bundle-monad},
    recplacing the integers in the indices by words.
    The proof that $\Sswapnm$ is a distributive law is the same proof as 
    \cref{prop:tangent-bundle-distributive}, upon observing that 
    $\Sswapnm[m][n] \compl \Sswapnm = \id$. 
\end{proof}

\begin{theorem} \label{thm:iterated-monad} 
    For all $n \in \N$, 
    $\Sinjzn : \idfun \naturalTrans \T^n$, 
    $\SmonadSumn : \T^n \T^n \naturalTrans \T$ and 
    $\Sswapnm : \T^n \T^m \naturalTrans \T^m \T^n$ are natural in $\cat$.
    Thus, $\T^n$ is a monad and $\Sswapnm$ is a distributive law 
    of $\T^n$ over $\T^m$.
\end{theorem}

\begin{proof}
    First, we prove that $\Sinjzn$ 
    is a natural transformations $\idfun \naturalTrans \T^n$.
    We can prove by a straightforward induction that 
    $\Sinjzn$ is the horizontal composition of 
    $\Sinjz : \idfun \naturalTrans \S$ in $\catLin$ with itself, 
    taken $n$ times 
    \[ \begin{tikzcd}[column sep = large]
	{\Sinjzn = \idfun \cdots \idfun} & {\S \cdots \S}
	\arrow["{\Sinjz \hcomp \cdots \hcomp \Sinjz}", from=1-1, to=1-2] 
    \end{tikzcd} \]
    But $\Sinjz$ is linear, so $\T \Sinjz = \S \Sinjz$. Thus, 
    $\Sinjz$ is also the horizontal composition 
    of $\Sinjz : \idfun \naturalTrans \T$ in $\cat$ with itself, taken $n$ times.
    Thus, it is a natural transformation $\idfun \naturalTrans \T^n$ in $\cat$.
    The proof of naturality of $\Sscalen r$ is the same.
    Then, we prove that $\Sswapnm$ is a natural transformation 
    $\T^n \T^m \naturalTrans \T^n \T^m$ for all $n, m$.
    We can prove by a straightforward computation that 
 \[ 
\begin{tikzcd}[row sep = abysmal]
	{\Sswapnm[1][m+1] = \S \S \S^m} & {\S \S \S^m} & {\S \S^m \S} \\
    {\Sswapnm[n+1][m] = \S \S^n \S^m} & {\S \S^m \S^n} & {\S^m \S \S^n,}
	\arrow["{\Sswap \S^m}", from=1-1, to=1-2]
	\arrow["{\S \Sswapnm[1][m]}", from=1-2, to=1-3]
	\arrow["{\S \Sswapnm}", from=2-1, to=2-2]
	\arrow["{\Sswapnm[1][m] \S^n}", from=2-2, to=2-3]
\end{tikzcd} \]
this corresponds to a decomposition of a cyclic permutation into a product of transpositions.
But we know that $\Sswapnm$ is $\d$-linear for all $n, m \in \N$, so 
we get that 
\[ \Sswapnm[1][n+1] = \T \Sswapnm[1][m] \comp (\Sswap \T^m) \qquad 
\Sswapnm[n+1][m] = (\Sswapnm[1][m] \T^n) \comp \T \Sswapnm . \]
We conclude that $\Sswapnm : \T^n \T^m \naturalTrans \T^m \T^n$ is a natural transformation 
by a straightforward induction, using the naturality of 
$\Sswap : \T \T \naturalTrans \T \T$.
Finally, we prove that $\SmonadSumn$ is a natural transformation 
$\T^n \T^n \naturalTrans \T^n$. Observe that 
\[ 
\begin{tikzcd}
	{\SmonadSumn[n+1] = \S \S^n \S \S^n } & {\S \S \S^n \S^n} & {\S \S^n \S^n} & {\S \S^n}
	\arrow["{\S \Sswapnm[n][1] \S^n}", from=1-1, to=1-2]
	\arrow["{\SmonadSum \S^n}", from=1-2, to=1-3]
	\arrow["{\S \SmonadSumn}", from=1-3, to=1-4]
\end{tikzcd} \]
By $\d$-linearity of $\Sswapnm[n][1]$ and $\SmonadSumn$, we get that 
$\SmonadSumn[n+1] = \T \SmonadSumn \comp (\SmonadSum \T^n) \comp (\T \Sswapnm[n][1] \T^n)$.
We conclude that $\SmonadSumn[n]$ is a natural transformation 
$\T^n \T^n \naturalTrans \T$ by a straightforward induction, using the 
naturality of $\SmonadSum$ and $\Sswapnm[1][n]$.
\end{proof}

\section{Taylor expansion as a functor} \label{sec:taylor-expansion}

In this section, we address the central question of the paper. We 
define an operator $\Tn$ that performs an
order $n$ Taylor expansion, and we prove its compositionality.
Remember that the Taylor expansion of a smooth function is given by 
\begin{equation} \label{eq:Taylor-analysis}
    f(x + \varepsilon u) = \sum_{k=0}^n \frac{\varepsilon^k}{\factorial{k}} 
    \deriven{f}{k} \cdot (u, \ldots, u) + o(\varepsilon^n) 
\end{equation}
Intuitively, we would like to define $\Tn$ such that 
the $i$-th component of $\Tn f(x, u_1, \ldots, u_n)$ 
provides the term of order $\varepsilon^i$ in the Taylor expansion of 
$f(x + \varepsilon u_1 + \cdots + \varepsilon^n u_n)$.
For example, $\Tn[2] f(x, u_1, u_2)$ should be equal to 
\begin{equation} \label{eq:Ttwo}
    \left(f(x), \derive{f} \cdot u_1, \frac{1}{2} \deriven f 2
\cdot (u_1, u_1) + \derive{f} \cdot u_2 \right).
\end{equation}
Formally, let us assume that $\cat$ is a cartesian differential category, and that its associated 
semiring $\rig$ has multiplicative inverses for integers (see \cref{def:multiplicative-inverse}).
For any $f \in \cat(X, Y)$, define $\Tn f \in \cat(\Sn X, \Sn Y)$ 
by the following equation:
\begin{equation} \label{eq:Tn-intuitive}
    \Sproj_j \comp \Tn f = \sum_{k=1}^j \sum_{\twolines{i_1, \ldots, i_k > 0 \text{ s.t.}}{i_1 + \cdots + i_k = j}}
\frac{1}{\factorial k} \hod{f}{k} \comp \prodtuple{\Sproj_0, \Sproj_{i_1}, \ldots, 
\Sproj_{i_k}}. \end{equation}
This functor corresponds to the push forward of higher order dual numbers \cite{Szirmay20}.
Intuitively, an element $e \in \Sn X$ is a degree 
$n$ polynomial over a formal variable $\formalvar$ such that 
$\formalvar^{n+1} \neq 0$. Then, $\Tn f$ maps this 
polynomial to the best polynomial approximation of 
$f(e_0 + (e-e_0))$ given by Taylor expansion.
\begin{example} \label{ex:taylor-wrel}
    We can check in the relational model (see \cref{ex:wrel})
    that the functor $\Tn$ maps $f \in \klWREL(A, B)$ to 
    a matrix 
    $\Tn f \in \klWREL(\uplus_{i = 0}^{n} A, \uplus_{i=0}^n B)
    \simeq \rig^{\mfin(A)^{n+1} \times \uplus_{i=0}^n B}$
    whose coefficients of degree $i$ are given by 
    \[ (\Tn f)_{(m_0, \ldots, m_n), (i, b)}
    =  \frac{\factorial{(m_0 + \cdots + m_n)}}{
    \factorial{m_0} \cdots \factorial{m_n}} f_{m_0 + \cdots + m_n, b} 
    \]
    if $1 \size{m_1} + \cdots + n \size{m_n} = i$, and are equal to $0$
    otherwise (here, $\size{m} = \sum_{a \in A} m(a)$ is the size of the multiset).
    As explained in Sec. 9.4.2 of \cite{EhrhardWalch25}, those coefficients
    correspond to a Taylor expansion on formal power series. 
\end{example}

\subsection{Alternative definition of Taylor expansion}

\label{sec:taylor-expansion-def}

The issue with \cref{eq:Tn-intuitive} is that it is hard to manipulate.
The computation of $\Tn(g \comp f)$ would intensively require 
the Fa\'a di Bruno formula, and the computation of $\Tn^2 f$
seems nightmarish (we want to show that $\Tn$ 
is a monad, this implies intensive uses of $\Tn^2 f$).
We solve this issue by relating $\Tn$ with the higher order 
total derivative $\d^n f$ rather than the higher order 
derivative $\hod{f} n$. 
For example, \cref{eq:d-two} and \cref{eq:Ttwo} imply that 
\begin{equation} \label{eq:Ttwo-alternative}
    \Sproj_2 \comp \Tn[2] f \comp \prodtuple{x, u, v} 
= \d \d f \comp \prodtuplenosize{x, u, \frac{1}{2} u, v}.
\end{equation}
We generalize this alternative definition to all $n$ 
and prove in~\cref{thm:tn-justification} that 
this coincides with \cref{eq:Tn-intuitive}.
This provides simpler equations on $\Tn$, that we use 
to prove that $\Tn$ is a functor using 
the functoriality of $\T^n$.
We use a similar proof technique in \cref{sec:taylor-expansion-monad}
in order to prove that $\Tn$ is a monad, using the monad
structure on $\T^n$ of \cref{sec:iterated-monad}. 
For any $n \in \N^*$, define two natural transformations
$\Snode, \Snodebis : \Sn \naturalTrans \S \Sn[n-1]$ in $\catAdd$ and $\catLin$ by  
\begin{equation}
    \begin{split}
    \Snodebis<X> &= \prodtuple{\prodtuple{\Sproj_0, \ldots, \Sproj_{n-1}}, 
    \prodtuple{\Sproj_1, \ldots, \Sproj_n}} \\
    \Snode<X> &= \prodtuple{\prodtuple{\Sproj_0, \ldots, \Sproj_{n-1}}, 
\prodtuple{\frac{1}{n} \Sproj_1, \ldots, \frac{n}{n} \Sproj_n}}.
    \end{split}
\end{equation}
The subscript $X$ is omitted when writing $\Snode<X>$ and $\Snodebis<X>$.
For any $n \in \N$ and $f \in \cat(X, Y)$, 
define $\tn f \in \cat(\Sn X, Y)$ and 
$\tnbis f \in \cat(\Sn X, Y)$ by induction on $n$: 
$\tn[0] f = \tnbis[0] f = f$ and 
\begin{equation} \label{eq:tn-inductive}
\begin{tikzcd}[row sep = abysmal]
	{\tn[n+1] f    = \Sn[n+1] X} & {\S \Sn X} & Y \\ 
    {\tnbis[n+1] f = \Sn[n+1] X} & {\S \Sn X} & Y.
	\arrow["{\Snode[n+1]}", from=1-1, to=1-2]
	\arrow["{\d \tn f}", from=1-2, to=1-3]
    \arrow["{\Snodebis[n+1]}", from=2-1, to=2-2]
    \arrow["{\d \tnbis f}", from=2-2, to=2-3]
    \end{tikzcd}
    \end{equation}
Observe that $\tn[1] f = \tnbis[1] f = \d f$.
The operator $\tnbis$ is a generalization to differential categories 
of the \emph{higher order directional derivative} of def. 1 of~\cite{Huang06}.
%
On the other hand, \cref{thm:tn-justification} proves
that $\tn$ is equal to the sum in \cref{eq:Tn-intuitive}.
For example, $\tn[2] f \comp \prodtuple{x, u, v} = 
\d \d f \comp \prodtuple{x, u, \frac{1}{2} u, v}$ coincides with
\cref{eq:Ttwo-alternative}. 
\begin{remark}
    Both $\Snode$ and $\Snodebis$ duplicate values. 
    As such, they do not exist in the 
    setting of~\cite{Walch23,EhrhardWalch25} in which sums are partial.
\end{remark}    
There is an alternative definition of $\tn$ and $\tnbis$ 
that refers directly to the iterated derivative 
$\d^n$.
Define for any $n \in \N$ two natural transformations
$\Stree, \Streebis : \Sn \naturalTrans \S^n$ in $\catLin$
 by induction on $n$: $\Stree<X>[0] = \Streebis<X>[0] = \id_{X} \in \cat(X, X)$
and
\begin{equation} \label{eq:Stree}
    \begin{tikzcd}[row sep = abysmal]
        {\Stree<X>[n+1]    = \Sn[n+1] X} & {\S \Sn X} & {\S^{n+1} X} \\
        {\Streebis<X>[n+1] = \Sn[n+1] X} & {\S \Sn X} & {\S^{n+1} X.}
        \arrow["{\Snode<X>[n+1]}", from=1-1, to=1-2]
        \arrow["{\S \Stree<X>}", from=1-2, to=1-3]
        \arrow["{\Snodebis<X>[n+1]}", from=2-1, to=2-2]
        \arrow["{\S \Streebis<X>}", from=2-2, to=2-3]
    \end{tikzcd}
\end{equation}
Again, the subscript $X$ is omitted when writing 
$\Stree$ and $\Streebis$. \Cref{fig:Stree-graphical} provides a
graphical description of $\Stree[3]$.
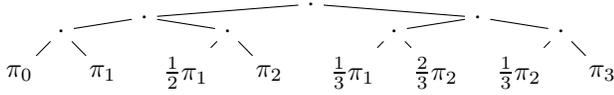
\begin{figure}
    \begin{center}
    \begin{tikzpicture}
        [level distance = 2em, 
        level 1/.style={sibling distance=0.5 \linewidth, level distance = 0.5em},
        level 2/.style={sibling distance=0.25 \linewidth, level distance = 0.5em},
        level 3/.style={sibling distance=0.125 \linewidth, level distance = 1.6em}]
        \node {.}
          child {node {.}
            child {node {.}
                child {node {$\Sproj_0$}}
                child {node {$\Sproj_1$}}}
            child {node {.}
                child {node {$\frac{1}{2} \Sproj_1$}}
                child {node {$\Sproj_2$}}}
          }
          child {node {.}
            child {node {.}
                child {node {$\frac{1}{3} \Sproj_1$}}
                child {node {$\frac{2}{3} \Sproj_2$}}}
            child {node {.}
                child {node {$\frac{1}{3} \Sproj_2$}}
                child {node {$\Sproj_3$}}}
          };
      \end{tikzpicture}
    \end{center}
    \caption{Graphical representation of $\Stree[3]$}
    \vspace{-1em}
    \label{fig:Stree-graphical}
\end{figure}
The following result is the counterpart
of Theorem 1 of~\cite{Huang06}.
\begin{proposition} \label{prop:tn-direct}
    $\tn f = \d^n f \comp \Stree$ and 
    $\tnbis f = \d^n f \comp \Streebis$
\end{proposition}
\begin{proof}
By induction on $n$. The base case $n = 0$ is trivial, 
and the inductive case consists of the following computation.
\begin{align*} 
    \d^{n+1}f \comp \Stree[n+1] 
&= \d \d^n f \comp \T \Stree \comp \Snode[n+1] 
\tag*{by linearity of $\Stree$} \\
&= \d (\d^n f \comp \Stree) \comp \Snode[n+1] 
\tag*{by \ref{def:cdc-chain}} \\
&= \d \tn f \comp \Snode[n+1]  
= \tn[n+1] f \tag*{I.H.} 
\end{align*}
The proof for the equation on $\tnbis f$ is the same.
\end{proof}

We now give an explicit computation 
of $\tn f$ that involves the higher order 
derivatives $\hod{f}{k}$ in order to ensure that $\tn f$ corresponds 
to \cref{eq:Tn-intuitive}.  
First, let us describe explicitly 
what are the components of $\Stree$. 
For any word $w \in \{0,1\}^*$, let 
$\factorial w = \prod_{i \in w} i$ (recall that 
$w$ can be seen as the set $\{i \st w_i = 1\}$).

\begin{proposition} \label{prop:Stree-explicit}
For any word $w$ of length $n$, 
\[ \Sproj_w \compl \Stree = \frac{\factorial{\weight{w}}}{w!} \Sproj_{\weight{w}} \quad \quad 
\Sproj_w \compl \Streebis = \Sproj_{\weight{w}}.\]
\end{proposition}
\begin{proof}
    We do the proof for $\Stree$, by induction on $n$. 
    The base case $n = 0$ is trivial (remember that $\Sn[0] = \idfun$).
    Let $w = v \cdot 0$, where $v$ is a word 
    of length $n$.
    \begin{align*} \Sproj_w \compl \Stree[n+1] 
    &= \Sproj_{v} \compl \Sproj_0 \compl \Stree[n+1] \tag*{by definition of $\Sproj_w$} \\
    &= \Sproj_{v} \compl \Stree[n] \compl \prodtuple{\Sproj_0, \ldots, \Sproj_n} \tag*{by definition
    of $\Stree$} \\
    &= \frac{\factorial{\weight{w}}}{\factorial w} \Sproj_{\weight w} 
    \tag*{by inductive hypothesis.}
    \end{align*}
    Now, let $w = v \cdot 1$, where $v$ is a word 
    of length $n$.
     \begin{align*} \Sproj_w \compl \Stree[n+1] &= 
    \Sproj_{v} \compl \Sproj_1 \compl \Stree[n+1] 
    = \Sproj_{v} \compl \Stree \compl \Sproj_1 \compl \Snode \tag*{by \cref{eq:Stree}} \\
    &= \frac{\factorial{\weight{v}}}{\factorial{v}} \Sproj_{\weight{v}} \compl 
    \Sproj_1 \compl \Snode  \tag*{by inductive hypothesis} \\
    &= \frac{\factorial{\weight{v}} \ (\weight{v} + 1)}{\factorial{v} \ (n+1)} \Sproj_{\weight{v} + 1}
    = \frac{\factorial{\weight w}}{\factorial{w}} \Sproj_{\weight w} 
     \end{align*}
     This concludes the proof.
    The proof for $\Streebis$ is a simpler version 
    of this proof, not involving coefficients.
\end{proof}

\begin{theorem} \label{thm:tn-justification}
    For all $f \in \cat(X, Y)$,
    \begin{equation} \label{eq:tn-justification}
        \tn f = \sum_{k=1}^n \sum_{\twolines{i_1, \ldots, i_k > 0\text{ s.t.}}{i_1 + \cdots + i_k = n}}
    \frac{1}{\factorial k} \hod{f}{k} \comp \prodtuple{\Sproj_0, \Sproj_{i_1}, \ldots, 
    \Sproj_{i_k}} \end{equation}
\end{theorem}

\begin{proof} We start the proof with the following computation.
\begin{align*}
\tn f &= \d^n f \comp \Stree \tag*{\cref{prop:tn-direct}} \\
&= \hspace{-2em} \sum_{(w^1, \ldots, w^k) \in \opart{\intsegment{n}}} \hspace{-0.5em} 
\frac{1}{\factorial k} \hod{f}{k} \comp 
\prodtuplenosize{\Sproj_{\zeroword}, \Sproj_{w^1}, \ldots, \Sproj_{w^k}} \comp \Stree
\tag*{\cref{prop:higher-order-derivative}}\\
&= \hspace{-2em} \sum_{(w^1, \ldots, w^k) \in \opart{\intsegment{n}}} 
\hspace{-0.5em}  \frac{1}{\factorial k} \hod{f}{k} \comp \prodtuplenosize{\Sproj_0, 
\frac{\factorial{\weight{w^1}}}{\factorial{w^1}} \Sproj_{\weight{w^1}}, \ldots, 
\frac{\factorial{\weight{w^k}}}{\factorial{w^k}} \Sproj_{\weight{w^k}}} \\
&= \hspace{-2em} \sum_{(w^1, \ldots, w^k) \in \opart{\intsegment{n}}} 
\hspace{-0.5em}  \frac{\factorial{\weight{w^1}} \ \cdots \ \factorial{\weight{w^k}}}{\factorial{n} \
\factorial{k}} 
\hod{f}{k} \comp  \prodtuplenosize{\Sproj_0, \Sproj_{\weight{w^1}}, \ldots, \Sproj_{\weight{w^k}}}
\end{align*}
We now use a well known result in combinatorics. 

\begin{fact} \label{fact:number-of-partitions}
For all $i_1, \ldots, i_k > 0$ such that $i_1 + \cdots i_k = n$,
the number of partitions $w^1, \ldots, w^k$ of $\intsegment{n}$
such that for all $j$, $\weight{w^j} = i_j$ is equal to the multinomial coefficient
$\frac{\factorial{n}}{\factorial{i_1} \ \cdots \ \factorial{i_k}}$.
\end{fact}
Thus, we can regroup terms of the sum based on the number of words 
$k$ and on the sizes $\weight{w^1}, \ldots, \weight{w^k}$ to conclude that
\[ \tn f = \sum_{k=1}^n
\sum_{\twolines{i_1, \ldots, i_k > 0 \text{ s.t.}}{i_1 + \cdots + i_k = n}}
\frac{1}{\factorial{k}} \hod{f}{k} \comp 
\prodtuplenosize{\Sproj_0, \Sproj_{i_1}, \ldots, \Sproj_{i_k}} . \]
This concludes the proof.
\end{proof}

\begin{remark}
Observe that by similar arguments
invoking \cref{prop:tn-direct,prop:higher-order-derivative,prop:Stree-explicit},
\begin{equation*} \label{eq:tnbis-justification-abelian-functor-calculus}
    \tnbis f = \hspace{-2em} \sum_{\{w^1, \ldots, w^k\} \in \upart{\intsegment{n}}} 
\hspace{-2em} \hod{f}{k} \comp  \prodtuplenosize{\Sproj_0, \Sproj_{\weight{w^1}}, 
\ldots, \Sproj_{\weight{w^k}}}
\end{equation*}
This equation is a clear counterpart of Thm 7.7 of~\cite{Bauer18} shown in the setting of the 
abelian functor calculus.
\end{remark}

For any $n, m \in \N$ such that $m \leq n$, we define
a natural transformation $\Snm : \Sn \naturalTrans \Sn[m]$ 
in $\catLin$ and $\catAdd$ by 
$\Snm = \prodtuple{\proj_0, \ldots, \proj_m}$.
For any $f \in \cat(X, Y)$ 
and any $n \in \N$, define two morphisms
$\Tn f$ and $\Tnbis f$ in $\cat(\Sn X, \Sn Y)$ by 
\begin{equation} \label{eq:Tn-definition}
    \Sproj_i \comp \Tn f = \tn[i] f \comp \Snm[n][i]
    \qquad 
    \Sproj_i \comp \Tnbis f = \tnbis[i] f \comp \Snm[n][i]
\end{equation}
By \Cref{thm:tn-justification}, this definition coincides 
with the definition in \cref{eq:Tn-intuitive}.

\subsection{Functoriality of the Taylor expansion}

The goal of this section is to prove the following result.
\begin{theorem} \label{thm:Tn-functor} 
    For any $n \in \N$, $\Tn, \Tnbis$ are functors, defining 
    the action on objects as $\Tn X = \Tnbis X = \Sn X$.
\end{theorem} 

The proof consists in two steps. First, we show that 
$\Snode$, $\Snodebis$, $\Stree$, $\Streebis$ and $\Snm$ 
satisfy equations that will eventually turn them 
(once we have shown the functoriality of $\Tn$ and $\Tnbis$)  
into natural transformations in $\cat$: \begin{enumerate}
    \item $\Snm : \Tn \naturalTrans \Tn[m]$ and 
    $\Snm : \Tnbis \naturalTrans \Tnbis[m]$.
    \item $\Snode : \Tn \naturalTrans \T \Tn[n-1]$ and 
$\Snodebis : \Tnbis \naturalTrans \T \Tnbis[n-1]$;
    \item $\Stree : \Tn \naturalTrans \T^n$ and
$\Streebis : \Tnbis \naturalTrans \T^n$;
\end{enumerate}
Then, we use those commutations and the functoriality of $\T^n$
to prove the functoriality of $\Tn$ and $\Tnbis$.

\begin{proposition} \label{prop:Snm-natural}
    The following diagrams commute.
    \[ \begin{tikzcd}[row sep = small]
        {\Sn X} & {\Sn[m] X} \\
        {\Sn Y} & {\Sn[m] Y}
        \arrow["\Snm", from=1-1, to=1-2]
        \arrow["{\Tn f}"', from=1-1, to=2-1]
        \arrow["{\Tn[m] f}", from=1-2, to=2-2]
        \arrow["\Snm"', from=2-1, to=2-2]
    \end{tikzcd}  \quad 
    \begin{tikzcd}[row sep = small]
        {\Sn X} & {\Sn[m] X} \\
        {\Sn Y} & {\Sn[m] Y}
        \arrow["\Snm", from=1-1, to=1-2]
        \arrow["{\Tnbis f}"', from=1-1, to=2-1]
        \arrow["{\Tnbis[m] f}", from=1-2, to=2-2]
        \arrow["\Snm"', from=2-1, to=2-2]
    \end{tikzcd}\] 
\end{proposition}

\begin{proof}
First, observe that for any $k \leq m \leq n$, 
$\Snm[m][k] \comp \Snm[n][m] = \Snm[n][k]$. 
Thus, for any $k \in \interval{0}{m}$,
\begin{multline*}
    \Sproj_k \comp \Tn[m] f \comp \Snm 
   = \tn[k] f \comp \Snm[m][k] \comp \Snm 
   = \tn[k] f \comp \Snm[n][k] \\
   = \Sproj_k \comp \Tn f  
   = \Sproj_k \comp \Snm \comp \Tn f
\end{multline*}
so we conclude that $\Tn[m] f \comp \Snm 
=  \Snm \comp \Tn f$. The proof that 
$\Tnbis[m] f \comp \Snm 
=  \Snm \comp \Tnbis f$ is the same.
\end{proof}

\begin{proposition} \label{prop:Snode-natural}
    The following diagrams commute.
    \[ \begin{tikzcd}[row sep = small]
        {\Sn X} & {\Sn[n-1] X} \\
        {\Sn Y} & {\Sn[n-1] Y}
        \arrow["\Snode", from=1-1, to=1-2]
        \arrow["{\Tn f}"', from=1-1, to=2-1]
        \arrow["{\T \Tn[n-1] f}", from=1-2, to=2-2]
        \arrow["\Snode"', from=2-1, to=2-2]
    \end{tikzcd} \quad 
    \begin{tikzcd}[row sep = small]
        {\Sn X} & {\Sn[n-1] X} \\
        {\Sn Y} & {\Sn[n-1] Y}
        \arrow["\Snodebis", from=1-1, to=1-2]
        \arrow["{\Tnbis f}"', from=1-1, to=2-1]
        \arrow["{\T \Tnbis[n-1] f}", from=1-2, to=2-2]
        \arrow["\Snodebis"', from=2-1, to=2-2]
    \end{tikzcd} \] 
    \end{proposition}
    \begin{proof}
    We prove the commutation on $\Snode$. Observe that
    \begin{align*} \Sproj_0 \comp \T \Tn[n-1] f \comp \Snode 
    &= \Tn[n-1] f \comp \Sproj_0 \comp \Snode 
    = \Tn[n-1] f \comp \Snm[n][n-1] \\
    \Sproj_0 \comp \Snode \comp \Tn f &= \Snm[n][n-1] \comp \Tn f \end{align*}
    so $\Sproj_0 \comp \T \Tn[n-1] f \comp \Snode = \Sproj_0 \comp \Snode \comp \Tn f$
    by \cref{prop:Snm-natural}. Then,
    \begin{align*} \Sproj_1 \comp \T \Tn[n-1] f \comp \Snode 
    &= \d \prodtuple{\tn[i] f \comp \Snm[n-1][i]}_{i=0}^{n-1} \comp \Snode \\
    &= \prodtuple{\d (\tn[i] f \comp \Snm[n-1][i])}_{i=0}^{n-1} \comp \Snode \tag*{\cref{prop:D-pairing}}\\
    &= \prodtuple{\d \tn[i] f \comp \T \Snm[n-1][i] \comp \Snode}_{i=0}^{n-1} \tag*{\ref{def:cdc-chain}}\\
    &= \prodtuple{\d \tn[i] f \comp \S \Snm[n-1][i] \comp \Snode}_{i=0}^{n-1}
    \end{align*}
    by $\d$-linearity of $\Snm[n-1][i]$.
    But, for any $i \neq 0$, \begin{align*}
    \S \Snm[n-1][i] \Snode &= \prodtuple{\Snm[n-1][i]
    \compl \Snm[n][n-1], \ \Snm[n-1][i] 
    \compl \prodtuple{\frac{1}{n} \Sproj_1, \ldots, \frac{n}{n} \Sproj_n} }\\\ 
    &=\prodtuple{\Snm[n][i], \prodtuple{\frac{1}{n} \Sproj_1, \ldots, \frac{i+1}{n} \Sproj_{i+1}}} \\
    &= \prodtuple{\Snm[n][i], \frac{i+1}{n} \prodtuple{\frac{1}{i+1} \Sproj_1, \ldots, 
    \frac{i+1}{i+1} \Sproj_{i+1}}}
    \end{align*}
    Thus, $\S \Snm[n-1][i] \compl \Snode = \Sscale{\frac{i+1}{n}} \compl \Snode[i+1]
    \compl \Snm[n][i+1]$.
    But then, 
    \begin{align*} 
        \d  \tn[i] f \comp \S \Snm[n-1][i] \comp \Snode 
    &= \d \tn[i] f \comp \Sscale{\frac{i+1}{n}} \comp \Snode[i+1] \comp \Snm[n][i+1] \\
    &= \frac{i+1}{n} \d \tn[i] f \comp \Snode[i+1] \comp \Snm[n][i+1] \tag*{\ref{def:cdc-left-additive}} \\
    &= \frac{i+1}{n} \tn[i+1] f \comp \Snm[n][i+1]
    \end{align*}
    and it follows that $\Sproj_1 \comp \T \Tn[n-1] f \comp \Snode 
    = \Sproj_1 \comp \Snode \comp \Tn f$.
    This concludes the proof for $\Snode$. The proof for $\Snodebis$ is similar.
    \end{proof}
    
    \begin{proposition} \label{prop:Stree-natural}
        The following diagrams commute.
        \[ \begin{tikzcd}[row sep = small]
            {\Sn X} & {\S^n X} \\
            {\Sn Y} & {\S^n}
            \arrow["\Stree", from=1-1, to=1-2]
            \arrow["{\Tn f}"', from=1-1, to=2-1]
            \arrow["{\T^n f}", from=1-2, to=2-2]
            \arrow["\Stree"', from=2-1, to=2-2]
        \end{tikzcd} \quad 
        \begin{tikzcd}[row sep = small]
            {\Sn X} & {\S^n X} \\
            {\Sn Y} & {\S^n}
            \arrow["\Streebis", from=1-1, to=1-2]
            \arrow["{\Tnbis f}"', from=1-1, to=2-1]
            \arrow["{\T^n f}", from=1-2, to=2-2]
            \arrow["\Streebis"', from=2-1, to=2-2]
        \end{tikzcd} \]
    \end{proposition}
    
    \begin{proof} By induction. The equation holds trivially for $\Stree[0]
        = \Streebis[0] = \id$. Now, observe that $\Stree[n+1] = 
        \S \Stree \compl \Snode[n+1] = \T \Stree \comp \Snode[n+1]$
        by $\d$-linearity of $\Stree$. We conclude that 
        the commutation on $\Stree[n+1]$ holds by a straightforward 
        diagram chase.
        \[ 
        \begin{tikzcd}[row sep = small]
            {\Sn[n+1] X} & {\S \Sn X} & {\S \S^n X} \\
            {\Sn[n+1] Y} & {\S \Sn Y} & {\S \S^n Y}
            \arrow["{\Snode[n+1]}", from=1-1, to=1-2]
            \arrow["{\Tn[n+1] f}"', from=1-1, to=2-1]
            \arrow["{\cref{prop:Snode-natural}}"{description}, draw=none, from=1-1, to=2-2]
            \arrow["{\T \Stree}", from=1-2, to=1-3]
            \arrow["\T \Tn f", from=1-2, to=2-2]
            \arrow["{\text{I.H.}}"{description}, draw=none, from=1-2, to=2-3]
            \arrow["{\T \T^n f}", from=1-3, to=2-3]
            \arrow["{\Snode[n+1]}"', from=2-1, to=2-2]
            \arrow["{\T \Stree}"', from=2-2, to=2-3]
        \end{tikzcd} \]
        The proof for $\Streebis$ is the same.
    \end{proof}

    \begin{proposition} \label{prop:Stree-monic}
    $\Stree$ and $\Streebis$ are monic: if $f, g \in \cat(X, \Sn Y)$ satisfy 
    $\Stree \comp g = \Stree \comp f$ or $\Streebis \comp g = \Streebis \comp f$ 
    then $g = f$.
    \end{proposition}
    
    \begin{proof}
    Assume that $\Stree g = \Stree f$. Let $i \in \interval{0}{n}$, and let 
    $w$ be any word such that $\weight{w} = i$.
    Then $\Sproj_w \comp \Stree \comp g = \Sproj_w \comp \Stree \comp f$.
    That is, $\frac{\factorial{\weight w}}{\factorial{w}} \Sproj_i \comp g 
    = \frac{\factorial{\weight w}}{\factorial{w}} \Sproj_i \comp f$ by 
    \cref{prop:Stree-explicit}.
    Thus, $\Sproj_i \comp f = \Sproj_i \comp g$ for all $i$ ($\rig$ has 
    multiplicative inverses for integers), so $f = g$. 
    The proof that $\Streebis$ is monic is similar.  
    \end{proof}

    \begin{proof}[Proof of \cref{thm:Tn-functor}]
        Let $f \in \cat(X, Y)$ and $g \in \cat(Y, Z)$.
        \begin{align*}
        \Stree \comp \Tn (g \comp f) &= \T^n(g \comp f) \comp \Stree 
        \tag*{by \cref{prop:Stree-natural}} \\
        &= \T^n g \comp \T^n f \comp \Stree \tag*{by functoriality of $\T^n$} \\
        &= \T^n g \comp \Stree \comp \Tn f  \tag*{by \cref{prop:Stree-natural}} \\
        &= \Stree \comp \Tn g \comp \Tn f \tag*{by \cref{prop:Stree-natural}}
        \end{align*}
        so $\Tn (g \comp f) = \Tn g \comp \Tn f$ by monicity of $\Stree$.
        We prove that $\Tn \id_X = \id_{\Tn X}$ in a similar way, 
        using that $\T^n \id_X = \id_{\T^n X}$.
        The proof that $\Tnbis$ is a functor is similar.
    \end{proof}

    \begin{remark} \label{rem:tn-functor}
    It immediately follows from~\cref{thm:Tn-functor} that for all 
        $f \in \cat(X, Y)$ and $g \in \cat(Y, Z)$, 
        \begin{equation} \label{eq:tn-composition}
            \tn(g \comp f) = \tn g \comp \Tn f, \quad 
            \tnbis(g \comp f) = \tnbis g \comp \Tnbis f 
        \end{equation}
    This equation is a clear counterpart of Theorem 3 of \cite{Huang06}.
    \end{remark}

\begin{proposition} \label{prop:Tn-linear}
    A morphism $f \in \cat(X, Y)$ is $\d$-linear if and only if 
    $\Tn f = \Sn f$.
    \end{proposition}
    
    \begin{proof}
    If $\Tn f = \Sn f$ then by \cref{eq:Tn-definition},  
    $\tn[1] f \comp \Snm[n][1] = f \comp \Sproj_1$.
    It immediately follows that $\d f = \tn[1] f = f \comp \Sproj_1$,
    so $f$ is $\d$-linear. 
    Conversely, if $f$ is $\d$-linear, then 
    $\d^n f = f \comp \Sproj_{\oneword}$ (by a straightforward induction).
    Thus, 
    \[ \tn f = \d^n f \comp \Stree = f \comp \Sproj_{\oneword} \comp \Stree 
    = f \comp \Sproj_n \]  
    and by \cref{eq:Tn-definition}, 
    $\Sproj_i \comp \Tn f = f \comp \Sproj_i
    = \Sproj_i \comp \Sn f$ so $\Tn f = \Sn f$.
\end{proof}

\section{The monadic structure of Taylor expansion}

\label{sec:taylor-expansion-monad}

We have proved the existence of a functor 
$\Tn$ that performs an order $n$ Taylor expansion 
in a compositional way. We now prove that $\Tn$ has a monad 
structure similar to the monad structure of $\T$.
Define natural transformations 
$\Sninjz : \idfun \naturalTrans \Sn$, $\SnmonadSum : \Sn^2 \naturalTrans \Sn$, 
$\Snmswap : \Sn \Sn[m] \naturalTrans \Sn[m] \Sn$, $\Snlift : \Sn \naturalTrans \Sn^2$ and 
$\Snscale{r} : \Sn \naturalTrans \Sn$ (for all $r \in \rig$) in $\catAdd$ and $\catLin$ 
as the tuplings characterized by the following equations:
\begin{equation} \label{eq:Tn-natural}
\begin{aligned}
    \Sproj_k \compl \Sninjz =& \kronecker k 0 \id &
    \Sproj_k \compl \Snscale r =& r^k \Sproj_k &  \\
    \Sproj_i \compl \Sproj_j \Snmswap =& \Sproj_j \compl \Sproj_i &
    \Sproj_i \compl \Sproj_j \Slift =& \kronecker i j \Sproj_i \\
\end{aligned}
\Sproj_k \compl \SnmonadSum = \! \! \sum_{i+j = k} \! \! \Sproj_i \compl \Sproj_j
\end{equation}
These natural transformations are the counterpart of 
the natural transformations defined in \cref{eq:cdc-natural-equations}.

\subsection{The natural transformations of Taylor expansion}

The goal of this section is to prove that $\Sninjz$, $\SnmonadSum$, 
$\Snmswap$, $\Snlift$ and 
$\Snscale{r}$ are also natural transformations on $\Tn$.
As mentioned in \cref{sec:taylor-expansion}, 
we avoid heavy combinatorics by making use of the canonical 
monad structure of $\T^n$. This proof is similar to the 
proof of functoriality. First, we show in \cref{prop:Stree-monad-morphism} 
that $\Stree : \Tn \naturalTrans \T^n$ 
satisfies equations that will eventually turn $\Stree$ 
(once we have shown that $\Tn$ is a monad) into a monad morphism~\cite{Maranda66} (also 
called monad transformer~\cite{Liang95}) from $\Tn$ to $\T^n$
that sends the distributive law $\Snmswap$ to $\Sswapnm$.
Then, we use these commutations and the naturality equations on 
$\T^n$ to prove the naturality equations
on $\Tn$.

\begin{proposition} \label{prop:Stree-monad-morphism}
The following diagrams commute.
\[
\begin{tikzcd}[row sep = small]
	\idfun & \idfun \\
	\Sn & {\S^n}
	\arrow[Rightarrow, no head, from=1-1, to=1-2]
	\arrow["{\Sninjz}"', from=1-1, to=2-1]
	\arrow["{\Sinjzn}", from=1-2, to=2-2]
	\arrow["\Stree"', from=2-1, to=2-2]
\end{tikzcd} 
\begin{tikzcd}[row sep = small]
	{\Sn \Sn} & {\S^n \S^n} \\
	\Sn & {\S^n}
	\arrow["{\Stree \hcomp \Stree}", from=1-1, to=1-2]
	\arrow["\SnmonadSum"', from=1-1, to=2-1]
	\arrow["\SmonadSumn", from=1-2, to=2-2]
	\arrow["\Stree"', from=2-1, to=2-2]
\end{tikzcd} \] 
\[ 
\begin{tikzcd}[row sep = small]
	\Sn & {\S^n} \\
	\Sn & {\S^n}
	\arrow["\Stree", from=1-1, to=1-2]
	\arrow["\Snscale r"', from=1-1, to=2-1]
	\arrow["\Sscalen r", from=1-2, to=2-2]
	\arrow["\Stree"', from=2-1, to=2-2]
\end{tikzcd} 
\begin{tikzcd}[row sep = small]
	{\Sn \Sn[m]} & {\S^n \S^m} \\
	{\Sn[m] \Sn} & {\S^m \S^n}
	\arrow["{\Stree \hcomp \Stree[m]}", from=1-1, to=1-2]
	\arrow["\Snmswap"', from=1-1, to=2-1]
	\arrow["{\Sswapnm}", from=1-2, to=2-2]
	\arrow["{\Stree[m] \hcomp \Stree}"', from=2-1, to=2-2]
\end{tikzcd}\]
The same diagrams commute when replacing $\Sn$ by $\Tn$.
\end{proposition}
\begin{proof}
    We prove the diagram on
    $\Sninjz$/$\Sinjzn$ as follows.
\[ \Sproj_w \compl \Stree \compl \Sninjz 
= \frac{\factorial{\weight{w}}}{\factorial{w}} \Sproj_{\weight w} \compl \Sninjz
= \kronecker{\weight{w}}{0} \id = \kronecker{w}{\zeroword} = \Sproj_w \compl \Sinjzn \]
The proof on
$\Snscale r$/$\Sscalen r$ is similar.
We now prove the diagram on $\Snmswap$/$\Sswapnm$.
Observe that for any words $u, v$, 
$\Sproj_u \hcomp \Sproj_v = \Sproj_v \compl \Sproj_u$, so the 
projections of the two path are
\begin{align*} (\Sproj_u \hcomp \Sproj_v) \compl (\Stree[m] \hcomp \Stree) \compl \Snmswap 
&=  \frac{\factorial{\weight{u}} \ \factorial{\weight{v}}}{\factorial{u} \ \factorial{v}}
(\Sproj_{\weight u} \hcomp \Sproj_{\weight v}) \compl \Snmswap \\
&= \frac{\factorial{\weight{u}} \ \factorial{\weight{v}}}{\factorial{u} \ \factorial{v}}
(\Sproj_{\weight v} \hcomp \Sproj_{\weight u}) \\
    (\Sproj_u \hcomp \Sproj_v) \compl \Sswapnm \compl  (\Stree \hcomp \Stree[m])
&= (\Sproj_v \hcomp \Sproj_u) \compl (\Stree \hcomp \Stree[m]) \\
&= \frac{\factorial{\weight{u}} \ \factorial{\weight{v}}}{\factorial{u} \ \factorial{v}}
(\Sproj_{\weight v} \hcomp \Sproj_{\weight u})
\end{align*}
Finally, we prove the diagram involving $\SnmonadSum$ and $\SmonadSumn$.
\[ \Sproj_w \compl \Stree \compl \SnmonadSum 
= \frac{\factorial{\weight w}}{\factorial w} \compl \Sproj_{\weight w} \SnmonadSum 
= \frac{\factorial{\weight w}}{\factorial w} \compl \sum_{i_1 + i_2 = \weight w} 
\Sproj_{i_1} \compl \Sproj_{i_2} \] 
\begin{align*}
    \Sproj_w \compl \SmonadSumn \compl (\Stree \hcomp \Stree) 
    &= \hspace{-1em} \sum_{(w^1, w^2) \in \opart w}  \hspace{-1em} (\Sproj_{w^2} \hcomp \Sproj_{w^1}) 
    \compl (\Stree \hcomp \Stree) \\
    &= \hspace{-1em} \sum_{(w^1, w^2) \in \opart w} \hspace{-1em}
    \frac{\factorial{\weight{w^1}} \ \factorial{\weight{w^2}}}{\factorial{w}} 
    \Sproj_{\weight{w^1}} \compl \Sproj_{\weight{w^2}}
\end{align*}
We conclude that the two operands are equal, upon observing that there are 
exactly $\frac{\factorial{\weight{w}}}{\factorial{i_1} \ \factorial{i_2}}$
partitions $(w^1, w^2)$ of $w$ such that $\weight{w^j} = i_j$ by \cref{fact:number-of-partitions}.

Finally, the same diagrams commute when replacing $\Sn$ by $\Tn$, by 
$\d$-linearity of $\Stree$ and \cref{prop:Tn-linear}.
\end{proof}

\begin{theorem} \label{thm:naturality-Tn} 
The following families of morphisms are natural transformations in $\cat$: 
$\Sninjz : \idfun \naturalTrans \Tn$, $\SnmonadSum : \Tn^2 \naturalTrans \Tn$, 
$\Snmswap : \Tn \Tn[m] \naturalTrans \Tn[m] \Tn$ and 
$\Snscale{r} : \Tn \naturalTrans \Tn$.
\end{theorem}

\begin{remark}
On the other hand, $\SnmonadSum$ is \emph{not} a natural 
transformation $\Tnbis \Tnbis \naturalTrans \Tnbis$ because 
of the lack of coefficients that erase redundancy. This can be checked 
by an explicit computation in $\SMOOTH$ for $n = 2$.
\end{remark}

\begin{proof}
We prove the naturality of $\Snscale r$.
\begin{align*}
    \Stree \comp \Tn f \comp \Snscale r &= \T^n f \comp \Stree \comp \Snscale r
    \tag*{naturality of $\Stree$} \\
    &= \T^n f \comp \Sscalen r \comp \Stree 
    \tag*{by \cref{prop:Stree-monad-morphism}} \\
    &= \Sscalen r \comp \T^n f \comp \Stree 
    \tag*{naturality of $\Sscalen r$} \\
    &= \Sscalen r \comp \Stree \comp \Tn f 
    \tag*{naturality of $\Stree$} \\ 
    &= \Stree \comp \Snscale r \comp \Tn f 
    \tag*{by \cref{prop:Stree-monad-morphism}}
\end{align*}
so by monicity of $\Stree$, $\Tn f \comp \Snscale r = \Snscale r \comp \Tn f$
and $\Snscale r$ is natural. Similarly,
\begin{align*}
& (\Stree[m] \hcomp \Stree) \comp \Tn[m] \Tn f \comp \Snmswap \\
&= \T^m \T^n f \comp (\Stree[m] \hcomp \Stree) \comp \Snmswap \tag*{naturality of $\Stree$} \\ 
&= \T^m \T^n f \comp \Sswapnm \comp (\Stree \hcomp \Stree[m]) 
\tag*{by \cref{prop:Stree-monad-morphism}} \\ 
&= \Sswapnm \comp \T^n \T^m f \comp (\Stree \hcomp \Stree[m]) 
\tag*{naturality of $\Sswapnm$} \\
&= \Sswapnm \comp (\Stree \hcomp \Stree[m]) \comp \Tn \Tn[m] f
\tag*{naturality of $\Stree$} \\
&= (\Stree[m] \hcomp \Stree) \comp \Snmswap \comp \Tn \Tn[m] f
\tag*{by \cref{prop:Stree-monad-morphism}} 
\end{align*}
so $\Tn[m] \Tn f \comp \Snmswap = \Snmswap \comp \Tn \Tn[m] f$
by monicity of $\Stree[m] \hcomp \Stree$ (the composition of two monos 
is also a mono) and $\Snmswap$ is natural.
The proof that $\Sinjz$ and $\SnmonadSum$ are natural follow a very similar pattern.
\end{proof}

Those equations rewrite as direct equations on $\tn$, in the same way that naturality 
equations on $\T$ rewrite as equation on the derivative $\d$.

\begin{cor} \label{prop:naturality-tn}
    For all $i \in \interval{0}{n}$ and $j \in \interval{0}{m}$, 
    \begin{equation} \label{eq:Tn-projections}
        \Sproj_i \compl \Sproj_j \comp \Tn[m] \Tn f = 
            \tn[j] \tn[i] f \comp (\Snm[m][j] \hcomp \Snm[n][i]) 
    \end{equation}
    Thus, $\tn$ satisfies the following equations: \begin{enumerate}
        \item $\tn f \comp \Sninjz = 0$;
        \item $\tn f \comp \SnmonadSum = \sum_{i=0}^n \tn[n-i] \tn[i] f \comp (\Snm[n][n-i] \hcomp 
        \Snm[n][i])$; \label{eq:tn-monad-sum}
        \item $\tn[n] \tn[m] f \comp \Snmswap = \tn[m] \tn f$;
        \item $\tn[n] f \comp \Snscale r = r^n \tn f$. 
        \end{enumerate} 
\end{cor}
\begin{proof}
    We prove \cref{eq:Tn-projections} first,
\begin{align*} 
    \Sproj_i \compl \Sproj_j \comp \Tn[m] \Tn f 
&= \Sproj_j \comp \Tn[m] \Sproj_i  \comp \Tn[m] \Tn f \tag*{$\d$-linearity of $\Sproj_i$} \\
&= \Sproj_j \comp \Tn[m] (\Sproj_i \comp \Tn f) \\
&= \Sproj_j \comp \Tn[m] (\tn[i] f \comp \Snm[n][i]) \\
&= \Sproj_j \comp \Tn[m] (\tn[i] f) \comp \T \Snm[n][i] \\
&= \tn[j] \tn[i] f \comp \Snm[m][j] \comp \T \Snm[n][i].
\end{align*}
Then, all of the equations above are simply the projections $\Sproj_n$ 
of the naturality equations on $\Tn$.
For example,
$\Sproj_i \comp \Tn f \comp \SnmonadSum
= \tn[i] f \comp \Snm[n][i] \comp \SnmonadSum 
= \tn[i] f \comp \SnmonadSum[i] \comp (\Snm[n][i] \hcomp \Snm[n][i])$, and 
\begin{align*}  
    \Sproj_i \comp \SnmonadSum \comp \Tn^2 f 
&= \sum_{k=0}^i \Sproj_{i-k} \compl \Sproj_k 
\comp \Tn^2 f \\
&=\sum_{k=0}^i \tn[i-k] \tn[k] f 
        \comp (\Snm[i][i-k] \hcomp \Snm[i][k]) \tag*{by \cref{eq:Tn-projections}}
\end{align*} 
so \cref{eq:tn-monad-sum} holds by naturality of $\SnmonadSum$  
and by straightforward equations on the $\Snm$. The other proofs are similar.
\end{proof}

We now prove that $\Snlift$ is also 
a natural transformation, so that 
$\Tn$ satisfies the same naturality equations as $\T$.
Using the same proof technique as above, 
we would like to define a natural transformation 
$\Sliftn : \S^n \naturalTrans \S^{2n}$ characterized by 
$\Sproj_u \compl \Sproj_v \compl \Sliftn = \kronecker{u}{v} \Sproj_u$
and to prove that
\[ 
\begin{tikzcd}[row sep = small]
	\Sn & {\S^n} \\
	{\Sn \Sn} & {\S^n \S^n}
	\arrow["\Stree", from=1-1, to=1-2]
	\arrow["\Snlift"', from=1-1, to=2-1]
	\arrow["\Sliftn", from=1-2, to=2-2]
	\arrow["{\Stree \hcomp \Stree}"', from=2-1, to=2-2]
\end{tikzcd} \]  
commutes. However, this diagram does not commute, 
we can check that 
$(\Sproj_u \hcomp \Sproj_v) \compl (\Stree \hcomp \Stree) \compl \Snlift 
\neq 0$ if and only if $\weight{u} = \weight{v}$, whereas 
$(\Sproj_u \hcomp \Sproj_v) \compl \Sliftn \compl \Stree \neq 0$ 
if and only if $u = v$.
Thus, we need to perform a more explicit (and combinatorial) computation, making use of the relationship
between $\tn$ and $\d^n$. 
The proof is done in
\version{sec. B of the long version of this paper \cite{Walch25-vlong}}
{\cref{sec:appendix:Slift-natural}}.

\begin{proposition} \label{prop:Slift-natural}
    For all $n, i, j$ such that $n \geq i, n \geq j$, 
    \begin{equation}  \label{eq:tn-Snlift}
    \tn[j] \tn[i] f \comp (\Snm[n][j] \hcomp \Snm[n][i]) \comp \Snlift 
    = \begin{cases} \tn[i] f \comp \Snm[n][i] \text{ if } i = j \\ 
    0 \text{ otherwise}
    \end{cases}
    \end{equation}
    and $\Snlift$ is a natural transformation $\Tn \naturalTrans \Tn^2$.
    \end{proposition}

\subsection{An infinitary Taylor expansion}

This section defines a functor $\Tinf$ performing 
Taylor expansion at all order.
Assume that for all object $X$ of $\cat$, there exists 
a countable cartesian product
$\Sinf X = \prod_{i=0}^{\infty} X$.
Let $\Sproj_k : \Sinf \naturalTrans \idfun$, as it is now standard. 
Let $\N_{\omega} = \N \cup \{\omega\}$.
All the natural transformations in $\catAdd$ and $\catLin$ 
defined in \cref{eq:Tn-natural}
generalize to $n, m \in \Ninf$. Similarly, let
$\Snm[\omega][n] : \Sinf \naturalTrans \Sn$
be the natural transformation in $\catLin$ and $\catAdd$ given by 
$\Snm[\omega][n] = \prodtuple{\Sproj_0, \ldots, \Sproj_n}$.
Observe that this yield a cone in $\cat$, $\catAdd$ and $\catLin$, 
see \cref{fig:taylor-cone}.
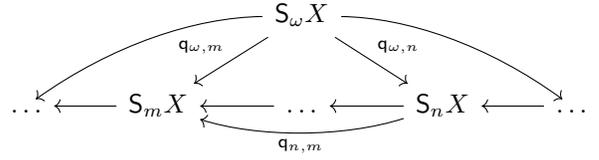
\begin{figure}
    \centering
\begin{tikzcd}[row sep = small]
	&& {\Sinf X} \\
	\\
	\ldots & {\Sn[m] X} & \ldots & {\Sn[n] X} & \ldots
	\arrow[curve={height=12pt}, from=1-3, to=3-1]
	\arrow["{\Snm[\omega][m]}"', from=1-3, to=3-2]
	\arrow["{\Snm[\omega][n]}", from=1-3, to=3-4]
	\arrow[curve={height=-12pt}, from=1-3, to=3-5]
	\arrow[from=3-2, to=3-1]
	\arrow[from=3-3, to=3-2]
	\arrow["\Snm", curve={height=-12pt}, from=3-4, to=3-2]
	\arrow[from=3-4, to=3-3]
	\arrow[from=3-5, to=3-4]
\end{tikzcd}
\caption{The cone of Taylor expansions}
\label{fig:taylor-cone}
\end{figure}

Let $\Tinf X = \Sinf X$, and for all $f \in \cat(X, Y)$, let 
$\Tinf f \in \cat(\Tinf X, \Tinf Y)$ 
be the unique morphism characterized by the equation 
$\Sproj_i \comp \Tinf f = \tn f \comp \Snm[\omega][n]$.

\begin{proposition} \label{prop:naturality-Tinf}
$\Tinf$ is a functor, $\Sproj_0 : \Tinf \naturalTrans \idfun$, 
$\Sinfinjz : \idfun : \naturalTrans \Tinf$, 
$\SinfmonadSum : \Tinf^2 \naturalTrans \Tinf$, $\Sinflift : \Tinf \naturalTrans \Tinf^2$, 
and $\Sinfscale r : \Tinf \naturalTrans \Tinf$
are natural transformations, and
$\Snmswap : \Tn \Tn[m] \naturalTrans 
\Tn[m] \Tn[m]$, $\Snm[\omega][n] : \Tinf \naturalTrans \Tn[n]$ 
are natural transformations
for all $n, m \in \Ninf$.
Finally, $f \in \cat(X, Y)$ is $\d$-linear if and only if $\Tinf f = \Sinf f$.
\end{proposition}

\begin{proof} The functoriality and the naturality equations on
    $\Tinf$ are equivalent, by projections, to equations on $\tn f$ for 
    all $n \in \N$. These equations are in turn equivalent to the functoriality 
    and the naturality equations on $\Tn$ for all $n \in \N$.
    For example, the functoriality of $\Tinf$ is equivalent to
    the equations of 
    \cref{eq:tn-composition}, for all $n \in \N$. These equations were proved
    as a consequence of the functoriality of $\Tn$.
    Similarly, the naturality of $\Sinfinjz$, 
    $\SinfmonadSum$, 
    $\Sinfscale r$ and $\Snmswap$ is a consequence of 
    \cref{prop:naturality-tn}, because the equations on $\tn$ given in
    \cref{prop:naturality-tn} are the projections of the
    naturality equations on $\Tinf$. 
    Similarly, the naturality of
    $\Sinflift : \Tinf \naturalTrans \Tinf^2$ is a consequence 
    of~\cref{prop:Slift-natural}.
    The proof that $\Snm[\omega][n]$ is natural is the same as the proof 
    of \cref{prop:Snm-natural}. Finally, the characterization of 
    the $\d$-linear morphisms is the same as the proof of \cref{prop:Tn-linear}.
\end{proof}

This functor $\Tinf$ is the counterpart of the coherent Taylor expansion $\T$
of~\cite{EhrhardWalch25},
more detail in~\cref{sec:infinitary-sums}.

\subsection{The monad structure of Taylor expansion}

In this section, we only consider indexes $n, m \in \Ninf$ 
such that $\Sn[n]$ is defined (so $n$ and $m$ can be equal to 
$\omega$, provided that the category has countable products).

\begin{proposition} \label{prop:Sn-monad}
$(\Sn, \Sninjz, \SnmonadSum)$ is a monad in $\catAdd$ and $\catLin$, 
and $\Snmswap$ is a distributive law of $\Sn$ over $\Sn[m]$.
\end{proposition}

\begin{proof}
Same proof as the proofs 
of~\cref{prop:tangent-bundle-monad,prop:tangent-bundle-distributive}, using that 
$\Snmswap \compl \Snmswap[m][n] = \id$.
\end{proof}

\begin{cor} \label{prop:Tn-monad}
$(\Tn ,\Sninjz, \SnmonadSum)$ is a monad in $\cat$ and 
    $\Snmswap$ is a distributive law of $\Tn$ over $\Tn[m]$.
\end{cor}

\begin{proof}
Direct consequence of \cref{prop:Sn-monad} and  \cref{prop:Tn-linear}.
\end{proof}

The Kleisli category $\kleisliTn$ of the monad $\Tn$
formalizes a notion of higher order dual numbers~\cite{Szirmay20}.
For all $n \in \N$,
a morphism $f \in \kleisliTn(X, Y)$ corresponds to a polynomial
$\sum_{k=0}^n f_k \formalvar^k$ of degree $n$ 
over a formal indeterminate 
$\formalvar$ such that $\formalvar^{n+1} = 0$.
When $n = \omega$, a morphism is a power series over a formal indeterminate 
$\formalvar$ that satisfies no equation.
The composition 
of $f$ with $g = \sum_{k=0}^n g_k \formalvar^k \in \kleisliTn(Y, Z)$ 
consists of performing a Taylor expansion 
of each $g_k$ around $f_0$ on the variation $f - f_0$, 
and returning the induced polynomial (or formal power series, when $n = \omega$) 
over $\formalvar$, taking into account that $\formalvar^{n+1} = 0$ if $n \in \N$.

As in \cref{sec:iterated-monad}, the distributive laws 
$\Snmswap$ satisfy the Yang Baxter equations, so 
by thm 2.1 of~\cite{Cheng11}, there exists a canonical monad structure 
on $\T_{i_1} \cdots \T_{i_n}$ for any $i_1, \ldots, i_n \in \Ninf$ 
as well as canonical distributive laws between those monads. 
The associated natural transformations are generalizations of 
the natural transformations 
given in \cref{eq:iterated-monad}, with 
projections taken on words on integers instead of words on $\{0,1\}$.
The intuition behind these monads is very similar to the intuition 
behind $\T^n$ given in \cref{sec:iterated-monad}.
Intuitively, an element $e \in \Tn[i_1] \cdots \Tn[i_n] X$ is a formal power series
over $n$ formal variables $\formalvar_1, \ldots, \formalvar_n$ 
such that for all $j$, $\formalvar_{j}^{i_j + 1} = 0$ (if $i_j = \omega$, then this 
is simply a formal variable, with no equation on it). The projection 
$\Sproj_{w_1 \cdots w_n}(e)$ gives the coefficient in $e$ of the monomial 
$\formalvar_1^{w_1} \cdots \formalvar_n^{w_n}$. 
Then, $\Tn[i_1] \cdots \Tn[i_n] f(e)$ pushes forward this power series
using the power series given by the Taylor expansion of $f$, taking 
into account that the monomials eventually vanish at high order,
since $\formalvar_{j}^{i_j} = 0$ for all $i_j \neq \omega$.

By $\d$-linearity of the projections 
$\proj_i$, the natural iso
\[ \SdistTimes = \prodtuple{\Sn \proj_1, \Sn \proj_2}
\in \catAdd \left( \Sn (X_1 \times X_2), \Sn X_1 \times \Sn X_2 \right)
\]
is also a natural isomorphism $\Tn (\_ \times \_) \naturalTrans \Tn(\_) \times \Tn(\_)$.
Then, both $\Sn$ and $\Tn$ equipped with $\SdistTimes$ are 
\emph{strong monoidal monads}, which turns them into  
\emph{commutative monads}~\cite{Kock70,Kock72}. Concretely, 
it means that there exist two natural transformations
$\SprodstrR$ and $\SprodstrL$ called \emph{strengths}
defined as 
\begin{equation} \label{eq:Tn-strength}
    \begin{split}
    \SprodstrR &= \SdistTimes^{-1} \compl \prodPair{\id}{\Sninjz}
    \in \cat(\Tn X_0 \times X_1, \Tn (X_0 \times X_1)) \\
    \SprodstrL &= \SdistTimes^{-1} \compl \prodPair{\Sninjz}{\id}
    \in \cat(X_0 \times \Tn X_1, \Tn (X_0 \times X_1)).
\end{split}
\end{equation}
These strengths induce two operators 
$\Tnpartial 1$ and $\Tnpartial 2$
that map $f \in \cat(X_0 \times X_1, Y)$ to
$\Tnpartial 0 f = \Tn f \comp \SprodstrR \in \cat (\Tn X_0 \times X_1, \Tn Y)$ 
and $\Tnpartial 1 f = \Tn f \comp \SprodstrL \in \cat(X_0 \times \Tn X_1, \Tn Y)$.
These operators perform a Taylor expansion of $f$ with regard to only 
one of its argument. The commutativity of the monad means that
\[ \SnmonadSum \comp \Tnpartial 0 \Tnpartial 1 f 
= \SnmonadSum \comp \Tnpartial 1 \Tnpartial 0 f 
= \Tn f \comp \SdistTimes . \]
It corresponds to the fact that the Taylor expansion of 
$f \in \cat(X_0 \times X_1, Y)$ can either be computed directly, or
by expanding $f$ with regard to its two arguments separately (in any order) by using the 
partial derivatives.
Note that those fundamentals categorical equations 
capture a complex underlying combinatorics.

\section{A direct axiomatization of Taylor expansion}

\label{sec:taylor-expansion-sound}

In this section, $\cat$ is a cartesian left $\semiring$-additive category
where $\rig$ has multiplicative inverses for integers. 
Let $n \in \Ninf$. 
We may assume 
$n = \omega$ or not, depending on the existence of
$\Sinf X$ in $\cat$. 
The work of the previous section suggest a 
direct categorical axiomatization of Taylor expansion. This 
axiomatization is very similar to def. 116 of~\cite{EhrhardWalch25}, 
except that their work considers a setting with infinite partial sums.

\begin{definition} \label{def:taylor-expansion-axiom} 
    An order $n$ Taylor expansion in $\cat$ is a functor $\Un$ such that 
    $\Un X = \Sn X$ and \begin{enumerate}
        \item $\Un \proj_i = \Sn \proj_i$, $\Un \Ssum = \Sn \Ssum$
        and $\Un \homothety r = \Sn \homothety r$; 
        \label{def:taylor-expansion-axiom-linearity}
        \item the families $\Sproj_0 : \Un \naturalTrans \idfun$, $\Sninjz : \idfun \naturalTrans \Un$, $\SnmonadSum : \Un^2 \naturalTrans \Un$,
        $\Snscale{r} : \Un \naturalTrans \Un$, 
        $\Snmswap[n][n] : \Un \Un \naturalTrans \Un \Un$ and 
        $\Snlift : \Un \naturalTrans \Un^2$ are natural transformations in $\cat$.
    \end{enumerate}
\end{definition}
\begin{proposition} \label{prop:taylor-canonical}
Any cartesian $\rig$-differential category such that $\rig$ has multiplicative 
inverses for integers has an order $n$ Taylor expansion, 
given by $\Tn$.
\end{proposition}

\begin{proof} $\Tn$ satisfies the equations of \cref{def:taylor-expansion-axiom-linearity}
by \cref{prop:Tn-linear}. $\Tn$ satisfies the naturality equations, 
by \cref{thm:naturality-Tn,prop:Slift-natural} when $n \in \N$, and by
\cref{prop:naturality-Tinf} when $n = \omega$. 
\end{proof}
We now show that \cref{def:taylor-expansion-axiom} only describes functors 
that are indeed Taylor expansions. More precisely, we prove 
that any order $n$ Taylor expansion $\Un$ induces a differential $\d$ and
that $\Un = \Tn$ where $\Tn$ is the Taylor expansion induced by $\d$.
First, we show that the term of degree $i$ in $\Un f$ only depends on 
the first $i$ coordinates.
For all $m \leq n$, let $\Smn : \Sn[m] \naturalTrans \Sn$ and 
$\Serase : \Sn \naturalTrans \Sn$ 
be natural transformations defined by 
\[\Sproj_k \compl \Smn = \begin{cases} 
    \Sproj_k \text{ if } k \leq m \\
    0 \text{ otherwise,}
\end{cases}
\quad \Sproj_k \compl \Serase = \begin{cases} 
    \Sproj_k \text{ if } k \leq m \\
    0 \text{ otherwise.}
\end{cases} \] 
Observe that $\Snm$ and $\Smn$ are sections and retractions since
$\Snm \compl \Smn = \id_{\Sn[m] X}$, and that $\Serase
= \Smn \compl \Snm$. 

\begin{lemma} \label{prop:Un-first-coordinates}
    For all Taylor expansion $\Un$, 
$\Sproj_i \comp \Un f = \Sproj_i \comp \Un f \comp \Serase[n][i]$.
\end{lemma}

The proof of this statement is in \version{sec. C of the long 
version of this paper \cite{Walch25-vlong}}{\cref{appendix:taylor-expansion-unique}}.
Let $\un[i] f = \Sproj_i \comp \Un f \comp \Smn[i][n] \in \cat(\Sn[i] X, Y)$.
By \cref{prop:Un-first-coordinates}, $\Sproj_i \comp \Un f = \un[i] f \comp \Snm[n][i]$, so 
$\un[i] f$ is to $\Un$ what $\tn[i] f$ is to $\Tn$. As such, the naturality 
equations on $\Un$ imply equations on $\un$ that are very similar to 
the equations of \cref{prop:naturality-tn}.

\begin{proposition} \label{prop:naturality-un}
    For all Taylor expansion $\Un$,
    \begin{equation} \label{eq:Un-projections}
        \Sproj_i \compl \Sproj_j \comp \Un \Un f = 
            \un[j] \un[i] f \comp (\Snm[n][j] \hcomp \Snm[n][i]) 
    \end{equation}
    Thus, the naturality equations on $\Un$ are equivalent to the following 
    equations: \begin{enumerate}
        \item $\un[0] f = f$ and $\un[i] f \comp \Sninjz[i] = 0$ for $i \geq 1$;
        \label{eq:un-monad-unit}
        \item $\un[i] f \comp \SnmonadSum[i] = \sum_{k=0}^i \un[i-k] \un[k] f 
        \comp (\Snm[i][i-k] \hcomp \Snm[i][k])$; \label{eq:un-monad-sum}
        \item $\un[i] \un[j] f \comp \Snmswap[i][j] = \un[j] \un[i] f$;
        \item $\un[i] f \comp \Snscale[i] r = r^i \un[i] f$;
        \item $\un[j] \un[i] f \comp (\Snm[n][j] \hcomp \Snm[n][i]) \comp \Snlift 
        = \begin{cases} \un[i] f \comp \Snm[n][i] \text{ if } i = j \\ 
            0 \text{ otherwise}
            \end{cases}$.
        \end{enumerate} 
    and $\un[1]$ is a derivative, in the sense of \cref{def:cdc}.
\end{proposition}

\begin{proof}
The proof of \cref{eq:Un-projections} is the same as the proof 
of \cref{eq:Tn-projections} in~\cref{prop:naturality-tn}. 
The other equations are simply the projections of the naturality
equations on $\Un$, we also refer to the proof of \cref{prop:naturality-tn}.
We conclude that $\un[1]$ is a derivative, 
taking those equations on $i = j =1$, and
by projections of the other 
equations of \cref{def:taylor-expansion-axiom} to recover 
\ref{def:cdc-additive}, \ref{def:cdc-projections} and \ref{def:cdc-chain}.
\end{proof}

Thus, there are two maps $\phi : \Un \mapsto \un[1]$ and 
$\psi : \d \mapsto \Tn$ between differentials and Taylor expansions, 
and by definition $\phi \comp \psi = \id$.
We want to prove that $\psi \comp \phi = \id$. 
First, we prove that this is the case when restricting $\phi$ on 
Taylor expansions that are "functorializations" of an actual Taylor expansion.

\begin{proposition} \label{prop:taylor-is-taylor-partial}
    The maps $\phi$ and $\psi$ induce a bijection between differentials
    and order $n$ Taylor expansions $\U$ such that 
    \begin{equation} \label{eq:taylor-is-taylor}
        \Un \comp \prodtuple{x, u, 0, \ldots} = 
    \prodtuple{\frac{1}{\factorial k} \hod{f}{k} \comp \prodtuple{x, u, \ldots, u}}_{k=0}^n
    \end{equation}
    where $\hod{f}{k}$ is the higher order derivative defined from the derivative 
    $\un[1] = \phi(\U)$. 
\end{proposition}

The proof is in \version{sec C. of the long version of 
this paper \cite{Walch25-vlong}}{\cref{appendix:taylor-expansion-unique}}.
We now prove that \cref{eq:taylor-is-taylor} always hold when the 
additive structure of $\cat$ is \emph{cancellative},
in the sense that $f + g = f + h$ implies $g = h$. We do not have a proof 
in the general case. 

\begin{remark}
 Unfortunately, categories in which the additive monoid is not cancellative
 include the categories with total sums that we will see in \cref{sec:infinitary-sums}.
 Those include the weighted relational model of \cref{ex:wrel}. Indeed, the completeness of a semiring $\rig$ 
induces the existence of an infinite coefficient 
$\infty = \sum_{x \in \rig} x$, and for all element $x$ $\infty + x = \infty$. 
Still, \cref{prop:taylor-is-taylor-partial} is already a satisfactory 
result, because it shows that the "functorialization" of a Taylor expansion is unique.;
\end{remark}

\begin{theorem} \label{thm:taylor-is-taylor} 
    Assume that the additive structure of $\cat$ is cancellative. 
    For all $n \in \Ninf$ such that $n \geq 1$, 
    there is a bijection between differentials $\d$ and order $n$ Taylor 
    expansions $\Un$, given by the maps 
    $\phi : \Un \mapsto \un[1]$ and 
    $\psi : \d \mapsto \Tn$.
\end{theorem}
The proof is in \version{sec. C of the long version of this 
paper~\cite{Walch25-vlong}}{\cref{appendix:taylor-expansion-unique}}.

\section{Taylor expansion and countable sums}

\label{sec:infinitary-sums}

We now assume that $\cat$ admits countable sums. 
Formally, this means that $\rig$ is a complete semiring,
that  $\cat(X, Y)$ is a complete monoid, 
and that the multiplicative action 
$\_ \cdot \_ : \rig \times \cat(X, Y) \arrow \cat(X, Y)$ respects 
countable sums in each variable.
We call such semimodules \emph{complete semimodules}.
A map between complete $\rig$-semimodules is $\rig$-additive if 
it preserves all countable sums, and the multiplicative action. 
The notions of cartesian left $\rig$-additive categories 
and cartesian differential categories directly carry to this 
new setting. We refer the reader to Sec. 6 of~\cite{Lemay24} 
for a more detailed definition.

\begin{definition}[Prop. 6.3 of \cite{Lemay24}] \label{def:taylor-category}
A cartesian differential category with countable sums is \emph{Taylor}
if $\rig$ has multiplicative inverses for integers and if 
for all $f \in \cat(X, Y)$ and $x, u \in \cat(Z, X)$,
\begin{equation} \label{eq:is-taylor}
f \comp (x + u) = f(x) + \sum_{k=1}^{\infty} \frac{1}{\factorial{k}} 
\hod f k \comp \prodtuple{x, u, \cdots, u}
\end{equation}
\end{definition}

One main examples of such Taylor categories is the
weighted relational model of \cref{ex:wrel}
over semirings with multiplicative 
inverses for integers. 
Typically, those categories are closed, and the differential is assumed to be compatible 
with the closed structure, see def. 4.4 of ~\cite{Bucciarelli10}.
As observed in~\cite{Lemay24}, a cartesian differential category with countable sums 
which is closed is Taylor if and only if it \emph{models Taylor expansion}, 
in the sense of def. 5.19 of~\cite{Manzonetto12}.
As such, those categories serve as models of the differential 
$\lambda$-calculus and its associated Taylor expansion~\cite{Ehrhard08,Manzonetto12}.


\begin{definition} \label{def:analytic}
    Let $\Ssum : \Sinf \naturalTrans \idfun$ be the natural transformation in $\catAdd$ defined 
    as $\Ssum = \sum_{k=0}^\infty \Sproj_k$.
    An order $\omega$ Taylor expansion $\Uinf$ is analytic if 
    $\Ssum$ is a natural transformation $\Uinf \naturalTrans \idfun$.
\end{definition}

The analytic Taylor expansions of def. 117 of~\cite{EhrhardWalch25}
coincide, in cartesian differential categories that admit countable sums,
with the analytic Taylor expansion of \cref{def:analytic}, hence the 
following result.

\begin{theorem} \label{thm:coherent-taylor-generalization}
A cartesian differential category is Taylor if and only if 
$\Tinf$ is an analytic Taylor expansion. Thus, any model of 
\cite{Manzonetto12} is a model of \cite{EhrhardWalch25}.
\end{theorem}

\begin{proof} Assume that $\Tinf$ is analytic. 
Then, 
$\Ssum \comp \Tinf \comp \prodtuple{x, u, 0, \ldots} 
= f \comp \Ssum \comp \prodtuple{x, u, 0, \ldots} 
= f \comp (x+u)$,
this is exactly \cref{eq:is-taylor} upon observing that 
$\Tinf f \comp \prodtuple{x,u, 0, \ldots} 
= \prodtuple{\frac{1}{\factorial{n}} \hod{f}{n} \comp 
\prodtuple{x, u, \ldots, u}}_{n=0}^{\infty}$.
Conversely, assume that $\cat$ is Taylor.
Then, 
\begin{align*}
f \comp \Ssum &= f \comp (\Sproj_0 + \Sproj_1) \comp \prodtuple{\Sproj_0, \sum_{k=1}^{\infty}
\Sproj_k} \\
&= \sum_{k=0}^{\infty} \frac{1}{\factorial{k}} \hod{f}{k} \comp 
\prodtuple{\Sproj_0, \sum_{k=1}^{\infty} \Sproj_k, \ldots, \sum_{k=1}^{\infty} \Sproj_k} 
\tag*{by \cref{eq:is-taylor}}
\end{align*}
By multilinearity of the derivative, we get that 
\[ f \comp \Ssum = \sum_{k=0}^{\infty} \sum_{a_1, \ldots, a_k \in \N^*} \frac{1}{\factorial{k}} 
\hod{f}{k} \comp \prodtuple{\Sproj_0, \Sproj_{a_1}, \ldots, \Sproj_{a_k}}. \]
By a reordering of the sum and \cref{thm:tn-justification}, this sum 
is equal to $\sum_{n=0}^{\infty} \tn f \comp \Snm[\omega][n] = \Ssum \comp \Tn f$.
So $\Ssum$ is natural and $\Tinf$ is analytic.
\end{proof}


%
%
%

\section{Conclusion}

This paper defined a functor $\Tn$ that performs an order $n$ Taylor expansion 
in a compositional way,
in any cartesian differential category (with multiplicative inverse for integers).
The fundamental properties of Taylor expansion then boil down to naturality 
equations that turns $\Tn$ into a monad. This monad provides a categorical
approach to higher order dual numbers~\cite{Szirmay20} and the jet bundle 
construction~\cite{Betancourt18,Huot22} 
recently developed in automated differentiation. 
Of particular interest is
the canonical monad $\Tn[i_1] \cdots \Tn[i_n]$, that freely
combines various Taylor expansion over different 
dual numbers, thanks to the theory of distributive laws.

A syntactical approach to this monad should be developed, with applications 
both in AD and in the differential lambda calculus. In particular, 
the coherent Taylor expansion of \cite{EhrhardWalch25} 
suggest that the coherent differential PCF of~\cite{Ehrhard22-pcf}
based on the tangent bundle $\T$ can be generalized to the infinitary
Taylor expansion $\Tinf$. 
This calculus would provide an alternative to the Taylor expansion of programs
that accounts for the determinism of computation, and
\cref{thm:coherent-taylor-generalization} ensures that this calculus 
would be a generalization of the previous one.

It should also be possible to unify this article with~\cite{EhrhardWalch25,Walch23}
in a single framework, by providing a more general theory of Taylor 
expansion with partial sums that encapsulate both the finite total sums of this 
article, and the partial \emph{positive} sums of~\cite{EhrhardWalch25}. A good candidate 
for the sums would be the notion of \emph{partial commutative monoids}~\cite{Hines13}.

This article also opens up a lot of perspective on the study of the categorical 
properties of Taylor expansion.
A generalization of the Taylor expansion functor to tangent categories~\cite{Cockett14}
should provide a categorical axiomatization of the jet construction. 
A generalization to reverse derivative~\cite{Cockett20} and reverse tangent 
categories~\cite{Crutwell24} would allow 
a computation of higher order derivatives and 
Taylor expansion in a reverse mode, which is very well suited for 
AD because of its efficiency. This should be similar to the notion 
of covelocity (the dual of jet bundles) mentionned in \cite{Betancourt18}.

\bibliographystyle{IEEEtran}
\bibliography{IEEEabrv, biblio.bib}
\appendices

\section{Proofs of \cref{sec:tangent-bundle}}
\label{appendix:tangent-bundle}

The goal of this appendix is to prove \cref{thm:naturality-Tn}, that 
we recall below. This proof was left 
as an appendix, because it mostly consists of an adaptation of the proof of 
Thm. 7 of~\cite{Walch23}.

\begin{theorem*} 
    Let $\d : \cat(X, Y) \arrow \cat(\S X, Y)$ be an operator
    that satisfies \ref{def:cdc-projections} and \ref{def:cdc-chain}.
    Then: 
    \begin{enumerate}
        \item $\d$ satisfies \ref{def:cdc-additive} iff $\T \Ssum = \S \Ssum$ and 
        $\T (\homothety r) = \S (\homothety r)$ for all $r \in \rig$ (that is, $\homothety r$ and $\Ssum$ are 
        $\d$-linear);
        \item $\d$ satisfies \ref{def:cdc-left-additive} iif 
        $\Sinjz : \idfun \naturalTrans \T$, $\SmonadSum : \T^2 \naturalTrans \T$
        and $\Sscale{r} : \T \naturalTrans \T$ (for all 
        $r \in \semiring$) are natural in 
        $\cat$\label{thm:tangent-bundle-natural-additive-annex};
        \item Assuming that $\Sinjz : \id \naturalTrans \T$ is natural in $\cat$, 
        $\d$ satisfies \ref{def:cdc-linear} iff $\Slift : \T \naturalTrans \T^2$
        is natural in $\cat$;
        \item $\d$ satisfies \ref{def:cdc-schwarz} iff $\Sswap : \T^2 \naturalTrans \T^2$ is natural
        in $\cat$.
    \end{enumerate}
    Thus, there is a bijection between derivative operators $\d$ and functors 
$\T$ such that $\T X = \S X$, $\T \proj_i = \S \proj_i$, 
$\T \Ssum = \S \Ssum$, $\T \homothety r = \S \homothety r$ and such that 
$\Sproj_0$, $\Sinjz$, $\SmonadSum$, $\Sscale{r}$ (for all $r \in \semiring$), 
$\Slift$ and $\Sswap$ are natural in $\cat$ with 
regard to $\T$.
\end{theorem*}

Let $\d : \cat(X, Y) \arrow \cat(X, \S Y)$ be an operator for all $X, Y$
that satisfies \ref{def:cdc-projections} and \ref{def:cdc-chain}.

\begin{proposition}
$\d$ satisfies \ref{def:cdc-additive} iff $\T \Ssum = \S \Ssum$ and 
$\T (\homothety r) = \S (\homothety r)$ for all $r \in \rig$ (that is, $\homothety r$ and $\Ssum$ are 
$\d$-linear).
\end{proposition}

\begin{proof}
Assume that $\d$ satisfies \ref{def:cdc-additive}. 
By \cref{prop:T-linear} it suffices to prove that 
$\d \Ssum = \Ssum \comp \Sproj_1$ and 
$\d \homothety r = \homothety r \comp \Sproj_1$ to conclude 
that $\T \Ssum = \S \Ssum$ and 
$\T (\homothety r) = \S (\homothety r)$.
By \ref{def:cdc-linear},
$\d \proj_i = \proj_i \comp \Sproj_1$.
Thus, by \ref{def:cdc-additive},
\[ \d \Ssum =  
\d \Sproj_0 + \d \Sproj_1
= \Sproj_0 \comp \Sproj_1 + \Sproj_1 \comp \Sproj_1 
= (\Sproj_0 + \Sproj_1) \comp \Sproj_1 = \Ssum \comp \Sproj_1. \]
\[ \d \homothety r = \d (r \cdot \id) = r \cdot \d \id = r \cdot \Sproj_1 
= \homothety r \comp \Sproj_1 \]
and it concludes the proof of the forward implication. 
Conversely, assume that $\Ssum$ and $\homothety r$ are linear.
Then,
\begin{alignat*}{3} 
    \d (f + g) &= \d (\Ssum \comp \prodPair{f}{g}) 
&&= \d \Ssum \comp \T \prodPair{f}{g} \tag*{by \ref{def:cdc-chain}}\\
&= \Ssum \comp \Sproj_1 \comp \T \prodPair{f}{g} 
&&= \Ssum \comp \d \prodPair{f}{g} \tag*{by assumption} \\
&= \Ssum \comp \prodPair{\d f}{\d g}
&&= \d f + \d g \tag*{by \cref{prop:D-pairing}} \\
    \d (r \cdot f) 
&= \d (\homothety r \comp f)  
&&= \d \homothety r \comp \T f \tag*{by \ref{def:cdc-chain}}\\
&= \homothety r \comp \Sproj_1 \comp \T f \tag*{by assumption}
&&= \homothety r \comp \d f  \\
&= r \cdot \d f 
\end{alignat*}
so $\d$ is $\rig$-additive.
\end{proof}

Before moving on to the other results, we describe a bit more precisely the 
components of $\T^2 f$.
\begin{lemma} \label{lemma:T-double-explicit}
    For all $f \in \cat(X, Y)$, 
        \[ \T^2 f = \prodPair{\prodPair{f \comp \Sproj_0 \compl \Sproj_0}{\d f \comp \Sproj_0}}
        {\prodPair{\d f \comp \T \Sproj_0}{\d \d f}} \]
\end{lemma}

\begin{proof} Two of the coordinates are given by a straightforward computation, 
    by definition of $\T$.
\[ \Sproj_0 \compl \Sproj_0 \comp \T^2 f = 
\Sproj_0 \comp \T f \comp \Sproj_0 = f \comp \Sproj_0 \compl \Sproj_0 \]
\[ \Sproj_1 \compl \Sproj_0 \comp \T^2 f 
= \Sproj_1 \comp \T f \comp \Sproj_0 = \d f \comp \Sproj_0 \] 
The other two relies on \ref{def:cdc-chain} and \ref{def:cdc-projections}.
\begin{align*} 
    \Sproj_1 \compl \Sproj_1 \comp \T^2 f 
&= \Sproj_1 \compl \S \Sproj_1 \comp \T^2 f  \tag*{naturality of $\Sproj_1$ in $\catLin$} \\
&= \Sproj_1 \comp \T \Sproj_1 \comp \T^2 f \tag*{by \ref{def:cdc-projections} and \cref{prop:T-linear}} \\
&= \Sproj_1 \comp \T(\Sproj_1 \comp \T f) \tag*{by \ref{def:cdc-chain}} \\
&= \Sproj_1 \comp \T(\d f) = \d \d f \tag*{by definition of $\T$} \\
\Sproj_0 \compl \Sproj_1 \comp \T^2 f 
&= \Sproj_1 \compl \S \Sproj_0 \comp \T^2 f  \tag*{naturality of $\Sproj_1$ in $\catAdd$} \\
&= \Sproj_1 \comp \T \Sproj_0 \comp \T^2 f \tag*{by \ref{def:cdc-projections} and \cref{prop:T-linear}} \\
&= \Sproj_1 \comp \T(\Sproj_0 \comp \T f) \tag*{by \ref{def:cdc-chain}} \\
&= \Sproj_1 \comp \T(f \comp \Sproj_0) \tag*{by definition of $\T$}\\ 
&= \d f \comp \T \Sproj_0  \tag*{by \ref{def:cdc-chain} and definition of $\T$} 
\end{align*}
\end{proof}

We now decompose the proof of \cref{thm:tangent-bundle-natural-additive-annex} into 
$3$ different subgoals.

\begin{lemma} \label{lemma:Sscale-natural}
   Let $r \in \rig$. The following assertions are equivalent: \begin{enumerate}
    \item $\d f \comp \prodPair{x}{r \cdot u} = r \cdot \d f \comp \prodPair{x}{u}$
for all $f \in \cat(X, Y)$ and $x, u \in \cat(Z, X)$
\label{lemma:Sscale-natural-equation};
\item $\d f \comp \Sscale r = r \cdot \d f$; 
    \item $\Sscale{r} : \T \naturalTrans \T$ is a natural transformation in $\cat$.
    \label{lemma:Sscale-natural-natural}
\end{enumerate}
\end{lemma}

\begin{proof} 
    $(1) \iff (2)$ Observe that $(2)$ is a special instance of $(1)$ in which 
    $x = \Sproj_0$ and $u = \Sproj_1$. Conversely, the composition of the equation of 
    $(2)$ by $\prodPair{x}{u}$ on the right yields the equation of $(1)$.

    $(2) \iff (3)$ Observe that
    \begin{align*}
\T f \comp \Sscale r &= \prodPair{f \comp \Sproj_0}{\d f \comp \Sscale r} \\
\Sscale r \comp \T f &= \prodPair{f \comp \Sproj_0}{r \cdot \d f} 
\end{align*} 
so $\Sscale r : \T \naturalTrans \T$ is natural if and only if 
$\d f \comp \Sscale r = r \cdot \d f$.
\end{proof}

\begin{lemma} \label{lemma:Sinjz-natural}
The following assertions are equivalent: \begin{enumerate}
    \item $\d f \comp \prodPair{x}{0} = 0$ for all $f \in \cat(X, Y), x \in \cat(Z, X)$;
    \item $\d f \comp \Sinjz = 0$;
    \item $\Sinjz$ is a natural transformation $\idfun \naturalTrans \T$.
\end{enumerate}
\end{lemma}
\begin{proof} $(1) \iff (2)$ Observe that $(2)$ is a special instance of $(1)$ in which 
    $x = \Sproj_0$. Conversely, the composition of the equation of 
    $(2)$ by $\prodPair{x}{0}$ on the right yields the equation of $(1)$.
    
    $(2) \iff (3)$ Observe that 
    \begin{align*}
        \T f \comp \Sinjz 
        &= \prodPair{f}{\d f \comp \Sinjz} \\
        \Sinjz \comp f &= \prodPair{f}{0}
        \end{align*}
so $\Sinjz : \idfun \naturalTrans \T$ is natural if and only if $\d f \comp \Sinjz = 0$.
\end{proof}

\begin{lemma} \label{lemma:SmonadSum-natural}
The following assertions are equivalent: \begin{enumerate}
    \item $\d f \comp \prodPair{x}{u + v} = \d f \comp \prodPair{x}{u} 
    + \d f \comp \prodPair{x}{v}$ for all $f \in \cat(X, Y)$ and $x, u, v \in \cat(Z, X)$;
    \item $\d f \comp \SmonadSum = \d f \comp \T \Sproj_0 
    + \d f \comp \Sproj_0$ for all $f \in \cat(X, Y)$ ;
    \item $\SmonadSum$ is a natural transformation 
    $\T^2 \naturalTrans \T$.
\end{enumerate}
\end{lemma}

\begin{proof}
    $(1) \iff (2)$ Observe that $(2)$ is an instance of $(1)$, taking 
    $x = \Sproj_0 \compl \Sproj_0$, $u = \Sproj_0 \compl \Sproj_1$ 
    and $v = \Sproj_1 \compl \Sproj_0$. Conversely, the composition of
    the equation of $(2)$ by $\prodPair{\prodPair{x}{u}}{\prodPair{v}{0}}$ 
    on the right yields the equation of $(1)$.

    $(2) \iff (3)$. Observe that 
\[ \Sproj_0 \comp \SmonadSum \comp \T^2 f 
= \Sproj_0 \compl \Sproj_0 \comp \T^2 f 
= f \comp \Sproj_0 \compl \Sproj_0 
= f \comp \Sproj_0 \comp \SmonadSum 
= \Sproj_0 \comp \T f \comp \SmonadSum \]
by \cref{lemma:T-double-explicit}. 
Thus, $\SmonadSum : \T^2 \naturalTrans \T f$ is natural if and only if 
\[ \Sproj_1 \comp \SmonadSum \comp \T^2 f 
= \Sproj_1 \comp \T f \comp \SmonadSum \] 
and we can check that this is equivalent to the equation of $(2)$. Indeed,
\begin{align*}
    \Sproj_1 \comp \SmonadSum \comp \T^2 f  
&= \Sproj_1 \compl \Sproj_0 \comp \T^2 f  
+ \Sproj_0 \compl \Sproj_1 \compl \T^2 f \\ 
&= \d f \comp \Sproj_0 + 
\d f \comp \T \Sproj_0 \tag*{by \cref{lemma:T-double-explicit}}  \\
\Sproj_1 \comp \T f \comp \SmonadSum &= \d f \comp \SmonadSum \tag*{by assumption}  
 \end{align*} 
so $(2)$ holds if and only if $\SmonadSum : \T^2 \naturalTrans \T$ is natural.
\end{proof}

\begin{proposition}
$\d$ satisfies \ref{def:cdc-left-additive} iif 
        $\Sinjz : \idfun \naturalTrans \T$, $\SmonadSum : \T^2 \naturalTrans \T$
        and $\Sscale{r} : \T \naturalTrans \T$ (for all 
        $r \in \semiring$) are natural in 
        $\cat$
\end{proposition}
\begin{proof} Direct consequence of 
    \cref{lemma:Sscale-natural,lemma:Sinjz-natural,lemma:SmonadSum-natural}.
\end{proof}

\begin{proposition}
    Assuming that $\Sinjz : \id \naturalTrans \T$ is natural in $\cat$, the following assertions 
    are equivalent: \begin{enumerate}
        \item $\d$ satisfies \ref{def:cdc-linear}, that is, 
        $\d \d f \comp \prodtuple{x, 0, 0, u} = \d f \comp \prodtuple{x, u}$
        for all $f \in \cat(X, Y)$ and $x, u \in \cat(Z, X)$;
        \item $\d \d f \comp \Slift = \d f$;
        \item $\Slift : \T \naturalTrans \T^2$ is a natural transformation in $\cat$.
    \end{enumerate}
\end{proposition}

\begin{proof}
$(1) \iff (2)$ Observe that $(2)$ is a special instance of $(1)$ in which 
$x = \Sproj_0$ and $u = \Sproj_1$. Conversely, the composition of the equation of 
$(2)$ by $\prodPair{x}{u}$ on the right yields the equation of $(1)$.

$(2) \iff (3)$ By \cref{lemma:T-double-explicit}
\[ \Sproj_1 \compl \Sproj_1 \comp \T^2 f \comp \Slift = \d \d f \comp \Slift 
\qquad \Sproj_1 \compl \Sproj_1 \comp \Slift \comp \T f = \Sproj_1 \compl \T f = \d f  \] 
so the equation of $(2)$ consists of the rightmost 
projection of the naturality equation on $\Slift$.
Thus, it suffices to show that 
\[ \Sproj_i \compl \Sproj_j \comp \T^2 f \comp \Slift = 
\Sproj_i \compl \Sproj_j \comp \Slift \comp \T f \]
for all $(i, j) \in \{0, 1\}^2 \setminus \{(1, 1)\}$ to conclude the equivalence.
We make use of \cref{lemma:T-double-explicit}.
\begin{itemize}
    \item Case $i = 0, j = 0$:
      $\Sproj_0 \comp \Sproj_0 \comp \Slift \comp \T f = \Sproj_0 \comp
      \T f = f \comp \Sproj_0$ and
      $\Sproj_0 \comp \Sproj_0 \comp \T^2 f \comp \Slift = f \comp
      \Sproj_0 \comp \Sproj_0 \comp \Slift = f \comp \Sproj_0$;
    \item Case $i = 1, j = 0$:
      $\Sproj_0 \comp \Sproj_1 \comp \Slift \comp \T f = 0 \comp \T f =
      0$ and
      $\Sproj_1 \comp \Sproj_0 \comp \T^2 f \comp \Slift
      = \d f \comp \Sproj_0 \comp \Slift 
      = \d f \comp \Sinjz \comp \Sproj_0 
      = 0$ by naturality of $\Sinjz$;
    \item Case $i = 0, j = 1$:
      $\Sproj_0 \comp \Sproj_1 \comp \Slift \comp \T f = 0 \comp \T f =
      0$ and
      $\Sproj_0 \comp \Sproj_1 \comp \T^2 f \comp \Slift = \d f
      \comp \T \Sproj_0 \comp \Slift = \d f \comp \Sinjz
      \comp \Sproj_0 =
      0$ by naturality of $\Sinjz$. 
  \end{itemize}
\end{proof}
For the last goal, we first prove useful equations on $\Sswap$.

\begin{lemma} \label{lemma:Sswap-explicit}
$\Sproj_i \compl \Sswap = \S \Sproj_i = \T \Sproj_i$ and 
$\T \Sproj_i \compl \Sswap = \S \Sproj_i \compl \Sswap = \Sproj_i$.
\end{lemma}

\begin{proof}
First, 
\[ \Sproj_j \compl \Sproj_i \compl \Sswap 
= \Sproj_i \compl \Sproj_j \quad 
\Sproj_j \compl \S \Sproj_j = \Sproj_j \compl \Sproj_i \]
so $\Sproj_i \compl \Sswap = \S \Sproj_i$.
Similarly, 
\[ \Sproj_j \compl \S \Sproj_i \compl \Sswap 
= \Sproj_i \compl \Sproj_j \compl \Sswap 
= \Sproj_j \compl \Sproj_i \]  
so $\S \Sproj_i \comp \Sswap = \Sproj_i$.
We conclude the proof by \cref{prop:T-linear} and \ref{def:cdc-linear}.
\end{proof}
\begin{proposition}
    The following assertions are equivalent: \begin{enumerate}
        \item $\d$ satisfies \ref{def:cdc-schwarz}, that is, 
        $\d \d f \comp \prodtuple{x, u, v, w} =
        \d \d f \comp \prodtuple{x, v, u, w}$ for all $f \in \cat(X, Y)$ and 
        $x, u, v, w \in \cat(Z, X)$;
        \item $\d \d f \comp \Sswap = \d \d f$;
        \item $\Sswap : \T^2 \naturalTrans \T^2$ is a natural transformation.
    \end{enumerate}
\end{proposition}

\begin{proof}
    $(1) \iff (2)$ Observe that $(2)$ is a special instance of $(1)$ in which 
$x = \Sproj_0 \compl \Sproj_0$, $u = \Sproj_0 \compl \Sproj_1$,
$v = \Sproj_1 \compl \Sproj_0$ and $w = \Sproj_1 \compl \Sproj_1$. 
Conversely, the composition of the equation of 
$(2)$ by $\prodPair{\prodPair{x}{u}}{\prodPair{v}{w}}$ 
on the right yields the equation of $(1)$.

$(2) \iff (3)$ By \cref{lemma:T-double-explicit},
\[  \Sproj_1 \compl \Sproj_1 \comp \T^2 f \comp \Sswap 
= \d \d f \comp \Sswap \quad 
\Sproj_1 \compl \Sproj_1 \comp \Sswap \comp \T^2 f 
= \Sproj_1 \compl \Sproj_1 \comp \T^2 f = \d \d f \]
so the equation of $(2)$ consists of the rightmost 
projection of the naturality equation on $\Sswap$.
Thus, it suffices to show that 
\[ \Sproj_i \compl \Sproj_j \comp \T^2 f \comp \Sswap = 
\Sproj_i \compl \Sproj_j \comp \Sswap \comp \T^2 f \]
for all $(i, j) \in \{0, 1\}^2 \setminus \{(1, 1)\}$ to conclude the equivalence.
We make use of \cref{lemma:T-double-explicit} and \cref{lemma:Sswap-explicit}.
    \begin{itemize}
    \item $i = 0, j = 0$: the computation is immediate;
    \item $i=1, j=0$:
    $\Sproj_1 \compl \Sproj_0 \comp \Sswap \comp \T^2 f  
    = \Sproj_0 \compl \Sproj_1 \comp \T^2 f = \d f \comp \T \Sproj_0$ and 
    $\Sproj_1 \comp \Sproj_0 \comp \T^2 f \comp \Sswap = 
      \d f \comp \Sproj_0 \comp \Sswap = \d f \comp \T \Sproj_0$;
    \item $i=0, j=1$: $\Sproj_0 \comp \Sproj_1 \comp \T^2 f \comp \Sswap
    = \d f \comp \T \Sproj_0 \comp \Sswap = \d f \comp \Sproj_0$ and
    $\Sproj_0 \compl \Sproj_1 \comp \Sswap \comp \T^2 f 
    = \Sproj_1 \comp \Sproj_0 \comp \T^2 f = \d f \comp \Sproj_0$.
    \end{itemize}
\end{proof}

This concludes the proof of \cref{thm:naturality-Tn}.

\section{Proof of \cref{prop:Slift-natural}}

\label{sec:appendix:Slift-natural}

\begin{lemma} \label{prop:double-tn-explicit}
    For all $i \in \interval{0}{n}$ and $j \in \interval{0}{m}$, we have
    $\tn[j] \tn[i] f = \d^{i+j} f \comp (\Stree[j] \hcomp \Stree[i])$.
\end{lemma}

\begin{proof} This is a direct consequence of \cref{prop:tn-direct} and 
    \cref{eq:faa-di-bruno-functor}: 
    $\tn[j] \tn[i] f = \d^j(\d^i f \comp \Stree[i]) \comp \Stree[j]
= \d^{i+j} f \comp \T^j \Stree[i] \comp \Stree[j]$.
\end{proof}

\begin{proposition}
    For all $n, i, j$ such that $n \geq i, n \geq j$, 
    \begin{equation} \label{eq:tn-Snlift-annex}
    \tn[j] \tn[i] f \comp (\Snm[n][j] \hcomp \Snm[n][i]) \comp \Snlift 
    = \begin{cases} \tn[i] f \comp \Snm[n][i] \text{ if } i = j \\ 
    0 \text{ otherwise}
    \end{cases}
    \end{equation}
    and $\Snlift$ is a natural transformation $\Tn \naturalTrans \Tn^2$.
\end{proposition}

\begin{proof}
By \cref{prop:higher-order-derivative},
\[ \d^{i+j} f = \hspace{-3em} \sum_{(u^1 v^1, \ldots, u^k v^k) \in 
\opart{\intsegment{i+j}}} 
\hspace{-3em} \frac{1}{\factorial k} \hod f k \comp \prodtuple{\Sproj_{\zeroword[i+j]}, \Sproj_{v^1} 
\hcomp \Sproj_{u^1}, \ldots, \Sproj_{v^l} \hcomp \Sproj_{u^l}} \]
where $\length{u^l} = i$ and $\length{v^l} = j$ for all $l$.
Thus, by \cref{prop:double-tn-explicit,prop:Stree-explicit}, 
\begin{multline*}
    \tn[j] \tn[i] f \comp (\Snm[m][j] \hcomp \Snm[n][i]) \comp \Snlift = 
    \hspace{-3em} \sum_{(u^1 v^1, \ldots, u^k v^k) \in \opart{\intsegment{i+j}}}
    \hspace{-3em} \frac{\left( \prod_{l=1}^k \factorial{\weight{u^l}} \right)
    \left( \prod_{l=1}^k \factorial{\weight{v^l}} \right)
    }{\factorial{i} \ \factorial{j} \ \factorial{k}}
    \\ \hod f k \comp \prodtuple{\Sproj_0, 
    \Sproj_{\weight{u^1}} \compl \Sproj_{\weight{v^1}}, \ldots, 
    \Sproj_{\weight{u^k}} \compl \Sproj_{\weight{v^k}}} \comp \Snlift
\end{multline*}
By definition of $\Snlift$ and by additivity of $\hod f k$ in its 
$k$ last variables, the terms of this sum vanish as soon as there 
is $l \in \intsegment{k}$
such that $\weight{u^l} \neq \weight{v^l}$. This is always the case when 
$i \neq j$, because $\sum_{j=1}^l \weight{u^l} = i \neq j = \sum_{l=1}^k \weight{v^l}$. 
So we get that 
\[ \tn[j] \tn[i] f \comp (\Snm[n][j] \hcomp \Snm[n][i]) \comp \Snlift 
= 0 \] 
for $i \neq j$. 
We now assume that $i = j$. Observe that
\begin{multline*}  
    \opart{\intsegment {i+i}} 
= \{(u^1 v^1, \ldots, u^k v^k) \st (u^1, \ldots, u^l) \in \opart{\intsegment{i}} \\
\text{ and } (v^1, \ldots, v^l) \in \opart{\intsegment{i}} \} 
\end{multline*}
By \cref{fact:number-of-partitions},
for any $(u^1, \ldots, u^k) \in \opart{\intsegment{i}}$, 
there is 
$\frac{\factorial{i}}{\factorial{\weight{u^1}} \ \cdots \factorial{\weight{u^k}}}$
elements $(v^1, \ldots, v^k) \in \opart{\intsegment i}$ such that 
$\weight{v^l} = \weight{u^l}$ for all $l \in \intsegment{k}$.
So we can regroup terms of the sum above to get that 
$\tn[i] \tn[i] f \comp (\Snm[m][i] \hcomp \Snm[n][i]) \comp \Snlift$
is equal to
\[ \hspace{-1em} \sum_{(u^1, \ldots, u^k) \in \opart{\intsegment i}} \hspace{-1em}
\frac{\prod_{l=1}^k \factorial{\weight{u^l}}}{\factorial{k} \ \factorial{i}}
\hod{f}{k} \comp \prodtuple{\Sproj_{0}, \Sproj_{\weight{u^1}}, \ldots, 
\Sproj_{\weight{u^k}}}. \]
This quantity is equal to $\tn[i] f$, as shown in the proof 
of \cref{thm:tn-justification}. This concludes the proof of 
\cref{eq:tn-Snlift-annex}. Finally, we prove that 
$\Snlift$ is natural. Observe that
\[ \Sproj_i \compl \Sproj_j \comp \Tn^2 \comp \Snlift 
= \tn[j] \tn[i] f \comp (\Snm[n][j] \hcomp \Snm[n][i]) \comp \Snlift . \] 
By \cref{eq:tn-Snlift-annex}, this is equal to 
$\kronecker i j \tn f \comp \Snm[n][i] 
= \Sproj_i \compl \Sproj_j \comp \Snlift \comp \Tn f$,
so $\Tn^2 f \comp \Snlift = \Snlift \comp \Tn f$ and 
$\Snlift$ is natural.
\end{proof}

\section{Proofs of \cref{sec:taylor-expansion-sound}}

\label{appendix:taylor-expansion-unique}

\begin{lemma*}
    For all Taylor expansion $\Un$, 
$\Sproj_i \comp \Un f = \Sproj_i \comp \Un f \comp \Serase[n][i]$.
\end{lemma*}

\begin{proof}
By induction on $i$, the case $i = 0$ is a direct consequence of the naturality of 
$\Sproj_0$. For the inductive hypothesis, assume that $i \geq 1$.
Then, let $h_i \in \cat(\Sn X, \Sn X)$ and $g_i \in \cat(\Sn X, \Sn^2 X)$ defined by 
\[ \Sproj_k \comp h_i = \begin{cases}
\Sproj_{k+1} \text{ if } k \in \interval{i}{n-1} \\ 
0 \text{ otherwise} 
\end{cases} g_i = \prodtuple{\Serase[n][i], h_i, 0 ,\ldots} \]
so that $\SnmonadSum \comp g_i = \id$.
Then, 
\[ \Sproj_i \comp \Un f = \Sproj_i \comp \Un f \comp \SnmonadSum \comp g_i 
= \Sproj_i \comp \SnmonadSum \comp \Un^2 f \comp g_i \]
by naturality of $\SnmonadSum$. Thus, 
\[ \Sproj_i \comp \Un f = \Sproj_i \compl \Sproj_0 \comp \Un^2 f \comp g_i
+ \sum_{k=1}^i \Sproj_{i-k} \compl \Sproj_k \comp \Un^2 f \comp g_i. \]
Observe now that $\Sproj_i \compl \Sproj_0 \comp \Un^2 f \comp g
= \Sproj_i \comp \Un f \comp \Serase[n][i]$ by naturality of 
$\Sproj_0$.
Furthermore, for all $k \in \interval{1}{n}$, 
\begin{align*} 
 &   \Sproj_{i-k} \compl \Sproj_k \comp \Un^2 f \comp g_i \\
&= \Sproj_k \comp \Un \Sproj_{i-k} \comp \Un^2 f \comp g_i \tag*{because $\Un \Sproj_{i-k} = \Sn \Sproj_{i-k}$} \\
&= \Sproj_k \comp \Un\Sproj_{i-k} \comp \Un^2 f \comp \Un \Serase[n][i-k] \comp g_i
\tag*{inductive hypothesis} \\
&= \Sproj_k \comp \Un\Sproj_{i-k} \comp \Un^2 f \comp  \Sn \Serase[n][i-k] \comp g_i 
\tag*{$\U \Serase[n][i-k] = \Sn \Serase[n][i-k] $} \\
&= \Sproj_k \comp \Un\Sproj_{i-k} \comp \Un^2 f \comp \Sninjz \comp \Serase[n][i-k] 
 \tag*{$\Serase[n][i-k] \comp h_i = 0$ ($k \geq 1$)} \\
&= \Sproj_k \comp \Sninjz \comp \Sproj_{i-k} \comp \Un f \comp \Serase[n][i-k] 
\tag*{naturality of $\Sninjz$} \\
&= 0 \tag*{$k \geq 1$}
\end{align*}    
using the fact that $\Serase[n][i-k] \comp h_i = 0$, since $k \geq 1$.
This conclude the inductive step: $\Sproj_i \comp \Un f = 
\Sproj_i \comp \Un f \comp \Serase[n][i]$.
\end{proof}

\begin{proposition*}
    The maps $\phi : \Un \mapsto \un[1]$ and 
    $\psi : \d \mapsto \Tn$ induce a bijection between differentials
    and order $n$ Taylor expansions $\U$ such that 
    \begin{equation} \label{eq:taylor-is-taylor-annex}
        \Un \comp \prodtuple{x, u, 0, \ldots} = 
    \prodtuple{\frac{1}{\factorial k} \hod{f}{k} \comp \prodtuple{x, u, \ldots, u}}_{k=0}^n
    \end{equation}
    where $\hod{f}{k}$ is the higher order derivative defined from the derivative 
    $\un[1] = \phi(\U)$. 
\end{proposition*}

\begin{proof}
    We have seen that $\phi \comp \psi = \id$. Let $\Un$ be an order $n$ Taylor expansion
    that satisfies \cref{eq:taylor-is-taylor-annex}, we prove that $\psi \comp \phi(\Un) = \Un$. 
    Let $\d = \un[1]$, and define $\tn$ and $\Tn$ from $\d$ as in \cref{sec:taylor-expansion}
    so that $\Tn = \psi(\phi(\Un))$.
 We want to prove that $\Un = \Tn $ so it suffices to prove that 
 $\un[i] f = \tn[i] f$ for all $f \in \cat(X, Y)$ and $i \leq n$. 
 We proceed by induction on $i$.
 
 The case $i = 0$ holds by naturality of $\Sproj_0$, and 
 the case $i = 1$ holds by definition, since $\tn[1] = \d = \un[1]$. Now, assume that 
 $i \geq 2$. Define
 $h_i \in \cat(\Sn[i] X, \Sn[i] X)$ as 
 $h_i = \prodtuple{0, \Sproj_2, \ldots, \Sproj_i, 0}$ 
 so that $\SnmonadSum[i] \prodtuple{\Serase[i][1], h_i, 0, \ldots, 0} 
 = \id$. Then, by \cref{eq:un-monad-sum,eq:un-monad-unit} of 
 \cref{prop:naturality-un},
 \begin{align*} 
     \un[i] f &= \un[i] f \comp \SnmonadSum[i] \comp \prodtuple{\Serase[i][1], h_i, 0, \ldots, 0} \\
     &= \un[i] f \comp \prodtuple{\Sproj_0, \Sproj_1, 0, \ldots, 0}
     + R_i
 \end{align*}
 with $R_i = \sum_{k=1}^{i-1} \un[k] \un[i-k] f \comp (\Snm[i][k] \hcomp \Snm[i][i-k]) 
 \comp \prodtuple{\Serase[i][1], h_i, 0, \ldots, 0}$.
 Similarly, by \cref{prop:naturality-tn}, 
 \[ \tn[i] f = \tn[i] f \comp \prodtuple{\Sproj_0, \Sproj_1, 0, \ldots, 0}
 + R_i' \] 
 with $R_i' = \sum_{k=1}^{i-1} \tn[k] \tn[i-k] f \comp (\Snm[i][k] \hcomp \Snm[i][i-k]) 
 \comp \prodtuple{\Serase[i][1], h_i, 0, \ldots, 0}$.
 By induction hypothesis, $R_i = R_i'$.  Furthermore,
 $\un[i] f \comp \prodtuple{\Sproj_0, \Sproj_1, 0, \ldots, 0} 
 = \tn[i] f \comp \prodtuple{\Sproj_0, \Sproj_1, 0, \ldots, 0}$ 
 by \cref{eq:taylor-is-taylor-annex}. So $\un[i] f = \tn[i] f$, this concludes the 
 inductive case. 
 \end{proof}

\begin{theorem*}
    Assume that the additive structure of $\cat$ is cancellative. 
    For all $n \in \Ninf$ such that $n \geq 1$, 
    there is a bijection between differentials $\d$ and order $n$ Taylor 
    expansions $\Un$, given by the maps 
    $\phi : \Un \mapsto \un[1]$ and 
    $\psi : \d \mapsto \Tn$.
\end{theorem*}
\begin{proof}
It suffices to prove that all order $n$ Taylor expansion $\Un$ satisfy 
\cref{eq:taylor-is-taylor-annex} in that setting. Let $\U$ be an order $n$ Taylor expansion.
Then, \cref{eq:taylor-is-taylor-annex} holds if and only if 
$\un[i] f \comp \prodtuple{\Sproj_0, \Sproj_1, 0, \ldots} 
= \tn[i] f \comp \prodtuple{\Sproj_0, \Sproj_1, 0, \ldots}$ for all 
$i \geq 1$, where $\tn[i]$ is defined 
from the derivative $\un[1]$ as in \cref{sec:taylor-expansion-def}.
We proceed by induction on $i$. The case $i = 1$ holds by definition 
of $\tn[1]$. 
By \cref{prop:naturality-un}, 
\[ \un[i] f \comp \prodtuple{\Sproj_0, 2 \Sproj_1, 0, \ldots} 
= 2^i \un[i] f \comp \prodtuple{\Sproj_0, \Sproj_1, 0, \ldots}\]
and, defining $g = \prodtuple{\Sproj_1, 0, \ldots}$, we 
have by \cref{eq:un-monad-sum}
\begin{align*} &\un[i] f \comp \prodtuple{\Sproj_0, 2 \Sproj_1, 0, \ldots}  \\
&= \un[i] f \comp \SnmonadSum[i] \comp
\prodtuple{\Serase[i][1], g, 0, \ldots} \\
&= 2 \un[i] f \comp \prodtuple{\Sproj_0, \Sproj_1, 0, \ldots} 
+ Q_i \end{align*} 
with $Q_i = \sum_{k=1}^{i-1} \un[k] \un[i-k] f \comp
(\Snm[i][k] \hcomp \Snm[i][i-k]) 
\comp \prodtuple{\Serase[i][1], g, 0, \ldots}$.
So, if the additive monoid is cancellative, it follows that 
\[ (2^n - 2) \un[i] f \comp \prodtuple{\Sproj_0, \Sproj_1, 0, \ldots, 0}
= Q_i . \]
Similarly, $(2^n - 2) \tn[i] f \comp \prodtuple{\Sproj_0, \Sproj_1, 0, \ldots, 0}
= Q_i'$
with $Q_i' = \sum_{k=1}^{i-1} \tn[k] \tn[i-k] f \comp
(\Snm[i][k] \hcomp \Snm[i][i-k]) 
\comp \prodtuple{\Serase[i][1], g, 0, \ldots}$.
By induction hypothesis, $Q_i = Q_i'$ so 
$\un[i] f \comp \prodtuple{\Sproj_0, \Sproj_1, 0, \ldots, 0} =
\tn[i] f \comp \prodtuple{\Sproj_0, \Sproj_1, 0, \ldots, 0}$
which concludes the proof.
\end{proof}

\end{document}